\documentclass[acmsmall,authorversion]{acmarttr}

\acmJournal{TOMS}

\usepackage{pdfsync}
\usepackage{verbatim}
\usepackage{graphicx}
\usepackage{xspace}
\usepackage{mathrsfs}
\usepackage{amsmath}
\usepackage{multirow}
\usepackage{rotating}
\usepackage{booktabs}
\usepackage{bm}

\def\..{\,\mathpunct{\ldotp\ldotp}} % Middle stuff for intervals. Usage: \..

\newcommand{\?}{\mskip1.5mu}
\newcommand{\sing}[1]{\left\{\?#1\?\right\}}
\newcommand{\Z}{\mathbf Z}
\newcommand{\lst}[2]{${#1}_0$,~${#1}_1$, $\dots\,$,~${#1}_{#2-1}$}

\newcommand{\xorshift}[1][]{\texttt{xorshift#1}\xspace}
\newcommand{\xoroshiro}[1][]{\texttt{xoroshiro#1}\xspace}
\newcommand{\xoroshirop}[1][]{\texttt{xoroshiro#1+}\xspace}
\newcommand{\xoroshiropp}[1][]{\texttt{xoroshiro#1++}\xspace}
\newcommand{\xoroshiros}[1][]{\texttt{xoroshiro#1*}\xspace}
\newcommand{\xoroshiross}[1][]{\texttt{xoroshiro#1**}\xspace}
\newcommand{\xorshifts}[1][]{\texttt{xorshift#1*}\xspace}
\newcommand{\xorshiftp}[1][]{\texttt{xorshift#1+}\xspace}
\newcommand{\xoshiro}[1][]{\texttt{xoshiro#1}\xspace}
\newcommand{\xoshirop}[1][]{\texttt{xoshiro#1+}\xspace}
\newcommand{\xoshiropp}[1][]{\texttt{xoshiro#1++}\xspace}

\newcommand{\xoshiross}[1][]{\texttt{xoshiro#1**}\xspace}

\newcommand{\mt}[1][]{\texttt{MT19937}\xspace}

\newcommand{\la}{\langle}
\newcommand{\ra}{\rangle}
\newcommand{\Det}{\operatorname{Det}}
\newcommand{\Detr}{\operatorname{\textrm{Det}^{\textrm{r}}}}
\newcommand{\Detc}{\operatorname{\textrm{Det}^{\textrm{c}}}}

\begin{document}

\bibliographystyle{ACM-Reference-Format-Journals}
\acmJournal{TOMS}
\title{Scrambled Linear Pseudorandom Number Generators}
\titlenote{This paper contains version 1.0 of the generators described
	therein. This work has been supported by a Google Focused Research Award. With respect to the published journal article,
	the present version corrects an equation in Proposition 7.3, and the number and maximum weight of $64$-bit \xoroshiro primitive polynomials for
	$2048$ and $4096$ bits of state in Tables~\ref{tab:xorpol} and~\ref{tab:xordeg} (thanks to Alan Zucconi for reporting the error).}
\author{David Blackman}
\affiliation{%
	\institution{Independent researcher}
	\country{Australia}
}
\author{Sebastiano Vigna}
\email{vigna@acm.org}
\orcid{0000-0002-3257-651X}
\affiliation{%
	\institution{Universit\`a degli Studi di Milano}
	\country{Italy}
}

\begin{abstract}

	$\mathbf F_2$-linear pseudorandom number generators are very popular due to their high
	speed, to the ease with which generators with a sizable state space
	can be created, and to their provable theoretical properties. However,
	they suffer from linear artifacts that show as failures in linearity-related
	statistical tests such as the binary-rank and the linear-complexity test.
	In this paper, we give two new contributions. First, we introduce two new
	$\mathbf F_2$-linear transformations that have been handcrafted to have good statistical
	properties and at the same time to be programmable very efficiently on superscalar
	processors, or even directly in hardware.
	Then, we describe some \emph{scramblers}, that is, nonlinear functions
	applied to the state array that reduce or delete the linear artifacts,
	and propose combinations of linear transformations and scramblers that give
	extremely fast pseudorandom number generators of high quality. A novelty in our
	approach is that we use ideas from the theory of filtered linear-feedback
	shift registers to prove some properties of our scramblers, rather than
	relying purely on heuristics. In the end, we provide simple, extremely
	fast generators that use a few hundred bits of memory, have provable
	properties, and pass strong statistical tests.
\end{abstract}




\maketitle

\section{Introduction}

In the last twenty years, in particular since the introduction of the Mersenne
Twister~\cite{MaNMT}, $\mathbf F_2$-\emph{linear}\footnote{Or, with an equivalent notation,
	$\Z/2\Z$-linear generators; since we will not
	discuss other types of linear generators, we will omit to specify
	the field in the rest of the paper.} pseudorandom number generators have been very popular:
indeed, they are often the stock generator provided by several programming
languages. Linear generators have several advantages: they are fast, it is easy
to create full-period generators with large state spaces, and thanks to their
connection with \emph{linear-feedback shift registers} (LFSRs)~\cite{KleSC} many
of their properties, such as full period, are mathematically provable. Moreover, if
suitably designed, they are rather easy to implement using simple xor and shift
operations.

The linear structure of such generators, however, is detectable by some
statistical tests for randomness: in particular, the binary-rank test~\cite{MarTMSRNS} and the
linear-complexity test~\cite{CarALC,ErdETBK} are failed by all linear
generators.\footnote{In principle: in practice, the specific instance of the
	test used must be powerful enough to detect linearity.}
Such tests are implemented, for example, by the testing framework TestU01~\cite{LESTU01}
under the name ``MatrixRank" and ``LinearComp'', respectively. These tests were indeed
devised to ``catch'' linear generators, and they are not considered problematic
by the community working on such generators, as the advantage of being able to
prove precise mathematical properties is perceived as outweighing the failure of
such tests (see~\cite{VigHTLGMT} for a more detailed discussion).

Nonetheless, one might find it desirable to mitigate or eliminate such \emph{linear
	artifacts} by \emph{scrambling} a linear generator, that is, applying a
nonlinear function to its state array to produce the actual output.
In this direction, two simple approaches are multiplication by a constant or
adding two components of the state array. However, while empirical tests usually
do not show linear artifacts anymore, the lower
bits are unchanged or just slightly modified by such operations. Thus, those
bits in isolation (or combined with a sufficiently small number of good bits)
will fail linearity tests.

In this paper, we try to find a middle ground by proposing very fast scrambled
generators with provable properties. By combining in different ways an
underlying linear engine and a scrambler we can provide different tradeoffs
in terms of speed, space usage, and statistical quality.

For example, \xoshiropp[256] is a $64$-bit generator with 256 bits of state that
emits a value in $0.86$\,ns on an Intel\textregistered{} Core\texttrademark{} i7-8700B CPU @$3.20$\,GHz (see
Table~\ref{tab:test64} for details); it passes all statistical tests we are
aware of, and it is $3$-dimensionally equidistributed. Multiple instances can be easily
parallelized using Intel's extended AVX2 instruction set, reducing the time to $0.30$\,ns (for eight instances).
Similarly,
\xoshiross[256] is $4$-dimensionally equidistributed, but it has a lower linear complexity.

However, if the user is interested in the generation of floating-point numbers
only, we provide a \xoshirop[256] generator that generates a value in $0.78$\,ns
(the value must then be converted to float); it is just $3$-dimensionally
equidistributed, and its lowest bits have low linear complexity, but since one needs
just the upper $53$ bits, the resulting floating-point values have no linear bias.
As in the previous case, instances can be parallelized, bringing down the time to $0.22$\,ns.

If space is an issue, a $\xoroshiropp[128]$, $\xoroshiross[128]$, or $\xoroshirop[128]$ generator
provides similar timings and properties in less space.
We also describe higher-dimensional generators, albeit mainly for theoretical
reasons, and $32$-bit generators with similar properties that are useful
for embedded devices and GPUs. Our approach can even provide fast, reasonable
$16$-bit generators.

Finally, we develop some theory related to
our linear engines and scramblers using results from the theory of noncommutative determinants
and from the theory of filtered LFSRs.

The C code for the generators described in this paper is available from the
authors and it is public domain.\footnote{\url{http://prng.di.unimi.it/}}
The test code is distributed under the GNU General Public License version 3 or later.

\section{Organization of the paper}

In this paper, we consider words of size $w$, $w$-bit operations, and generators
with $kw$ bits of state, $k\geq 2$. We aim mainly at $64$-bit generators (i.e., $w=64$), but we
also provide $32$-bit combinations.

The paper is organized in such a way to make immediately available code
and basic information for our new generators as quickly as possible: all
theoretical considerations and analyses are postponed to the second part of the
paper, albeit sometimes this approach forces us to point at subsequent material.

Our generators consist of a \emph{linear engine}\footnote{We use consistently
	``engine'' throughout the paper instead of ``generator'' when discussing
	combinations with scramblers to avoid confusion between the underlying linear generator
	and the overall generator, but the two terms are otherwise equivalent.}
and a \emph{scrambler}.
The linear engine is a linear transformation on $\Z/2\Z$, representable by a
matrix, and it is used to advance the internal state of the generator. The
scrambler is an arbitrary function on the internal state which computes the
actual output of the generator. We will usually apply the scrambler to the
current state, to make it easy for the CPU to parallelize internally the
operations of the linear engine and of the scrambler. Such a combination is
quite natural: for example, it was advocated by Marsaglia for \xorshift
generators~\cite{MarXR}, by Panneton and L'Ecuyer in their
survey~\cite{EcPFLRNG}, and it has been used in the design of
XSAdd~\cite{SaMXA} and of the Tiny Mersenne Twister~\cite{SaMTMT}.
An alternative approach is that of combining an $\mathbf F_2$-linear generator
with a linear congruential generator with large prime modulus~\cite{LECGCGCDF}.

In Section~\ref{sec:lin} we introduce our linear engines. In
Section~\ref{sec:scr} we describe the scramblers we will be using and their
elementary properties. Finally, in Section~\ref{sec:comb} we describe
generators given by several combinations between scramblers and linear
engines, their speed and their results in statistical tests.
Section~\ref{sec:choice} contains a guide to the choice of an appropriate generator.

In
Section~\ref{sec:poly} and~\ref{sec:ed} we discuss the mathematical properties of
our linear engines:
in particular, we introduce the idea of \emph{word polynomials}, polynomials on
$w\times w$ matrices associated with a linear engine. The word polynomial
makes it easy to compute the \emph{characteristic polynomial}, which is the basic
tool to establish full period. We then provide
\emph{equidistribution} results.

In the last part of the paper, starting with Section~\ref{sec:thscr}, we apply
ideas and techniques from the theory of filtered LFSRs to the problem of
analyzing the behavior of our scramblers.
We provide some exact results and discuss a few heuristics based on extensive
symbolic computation.
Our discussion gives a somewhat more rigorous foundation to the choices made in
Section~\ref{sec:comb}, and opens several interesting problems.

\section{Linear engines}
\label{sec:lin}

In this section we introduce our two linear engines \xoroshiro (xor/rotate/shift/rotate) and
\xoshiro (xor/shift/rotate).
All modern C/C++ compilers can compile a simulated rotation into a single CPU
instruction, and Java provides intrinsified rotation static methods to the same
purpose. As a result, rotations are no more expensive than a shift, and they
provide better state diffusion, as no bit of the operand is
discarded.\footnote{Note that at least one shift is necessary, as rotations and
	xors map the set of words $x$ satisfying $xR^s=x$ for a fixed $s$ into itself,
	so there are no full-period linear 	engines using only rotations.}

We denote with $S$ the $w\times w$ matrix on $\Z/2\Z$ that
effects a left shift of one position on a binary row vector (i.e., $S$ is all
zeroes except for ones on the principal subdiagonal) and with $R$ the
$w\times w$ matrix on $\Z/2\Z$ that effects a left rotation of one position (i.e., $R$ is
all zeroes except for ones on the principal subdiagonal and a one in the upper
right corner). We will use $\rho_r(-)$ to denote left rotation by $r$ of a
$w$-bit vector in formulae; in code, we will write
\texttt{rotl(-,r)}.

\subsection{{\fontsize{11pt}{11pt}\selectfont\xoroshiro}}

The \xoroshiro linear transformation updates cyclically two words of a larger state array.
The update rule is designed so that data flows through two computation paths of
length two with a single common dependency halfway, leading to good parallelizability
inside superscalar CPUs.

The base \xoroshiro linear transformation is obtained combining a
rotation, a shift, and again a rotation (hence the name), and it is defined by
the following $2w\times 2w$ matrix:
\[
	\mathscr X_{2w}=\left(\begin{matrix}
			R^a + S^b + I & R^c \\
			S^b + I       & R^c \\
		\end{matrix}\right).
\]
The general $kw\times kw$ form is given instead by
\begin{equation}
	\label{eq:M}
	\mathscr X_{kw}=\left(\begin{matrix}
			0      & 0      & \cdots & 0      & R^a + S^b + I & R^c    \\
			I      & 0      & \cdots & 0      & 0             & 0      \\
			0      & I      & \cdots & 0      & 0             & 0      \\
			\cdots & \cdots & \cdots & \cdots & \cdots        & \cdots \\
			0      & 0      & \cdots & I      & 0             & 0      \\
			0      & 0      & \cdots & 0      & S^b + I       & R^c
		\end{matrix}\right)
\end{equation}
Note that the general form applies the basic form to the first and last words
of state, and uses the result to replace the last and next-to-last words. The
remaining words are shifted by one position.
% Using the fact that the determinant of a block-diagonal matrix is the product
% of the blocks along the diagonal, we have
% \begin{align*}
% \det&\bigl(\mathscr X_{kw}\bigr) \\
% &=\det\left(\begin{matrix}
% 0 & 0 & \cdots & 0 & R^a + S^b + I & R^c\\
% I & 0 &  \cdots & 0 & 0 & 0\\
% 0 & I &  \cdots & 0 & 0 & 0\\
% \cdots&\cdots&\cdots&\cdots&\cdots&\cdots\\
% 0 & 0 &  \cdots & I & 0 & 0\\
% 0 & 0 &  \cdots & 0 & S^b + I & R^c
% \end{matrix}\right)\\
% &=\det\left(\begin{matrix}
% I & 0 &  \cdots & 0 & 0 & 0\\
% 0 & I &  \cdots & 0 & 0 & 0\\
% \cdots&\cdots&\cdots&\cdots&\cdots&\cdots\\
% 0 & 0 &  \cdots & I & 0 & 0\\
% 0 & 0 & \cdots & 0 & R^a  & 0\\
% 0 & 0 &  \cdots & 0 & S^b + I & R^c
% \end{matrix}\right)\\
% &=\det\bigl(R^a\bigr)\det\bigl(R^c\bigr)=1,
% \end{align*}
% so the \xoroshiro transform is invertible.

The structure of the transformation may appear repetitive, but it has been
so designed because this implies a very simple and efficient computation path.
Indeed, in Figure~\ref{fig:xoroshiro128} we
show the C code implementing the \xoroshiro transformation for $w=64$ with $128$
bits of state.
The constants prefixed with
``result'' are outputs computed using different scramblers, which will be
discussed in Section~\ref{sec:scr}.
The general case is better implemented using a form of cyclic update, as shown
in Figure~\ref{fig:xoroshiro1024}.

The reader should note that after the first xor, which represents the only data
dependency between the two words of the state array, the computation of the
two new words can continue in parallel, as depicted graphically in
Figure~\ref{fig:depxoroshiro128}.

\begin{figure}
	%\small
	\begin{verbatim}
    const uint64_t s0 = s[0];
    uint64_t s1 = s[1];

    const uint64_t result_plus = s0 + s1;
    const uint64_t result_plusplus = rotl(s0 + s1, R) + s0;
    const uint64_t result_star = s0 * S;
    const uint64_t result_starstar = rotl(s0 * S, R) * T;

    s1 ^= s0;
    s[0] = rotl(s0, A) ^ s1 ^ (s1 << B);
    s[1] = rotl(s1, C);
\end{verbatim}
	\caption{\label{fig:xoroshiro128}The C code for a
		\xoroshirop[128]/\xoroshiropp[128]/\xoroshiros[128]/\xoroshiross[128] generator.
		The array \texttt{s} contains two $64$-bit unsigned integers, not all zeros.}
\end{figure}

\begin{figure}
	%\small
	\begin{verbatim}
    const int q = p;
    const uint64_t s0 = s[p = (p + 1) & 15];
    uint64_t s15 = s[q];

    const uint64_t result_plus = s0 + s15;
    const uint64_t result_plusplus = rotl(s0 + s15, R) + s15;
    const uint64_t result_star = s0 * S;
    const uint64_t result_starstar = rotl(s0 * S, R) * T;

    s15 ^= s0;
    s[q] = rotl(s0, A) ^ s15 ^ (s15 << B);
    s[p] = rotl(s15, C);
\end{verbatim}
	\caption{\label{fig:xoroshiro1024}The C code for a
	\xoroshirop[1024]/\xoroshiropp[1024]/\xoroshiros[1024]/\xoroshiross[1024] generator. The state
	array \texttt{s} contains sixteen $64$-bit unsigned integers, not all zeros,
	and the integer variable $p$ holds a number in the interval $[0\..16)$.}
\end{figure}

\begin{figure}
	\centering
	\includegraphics{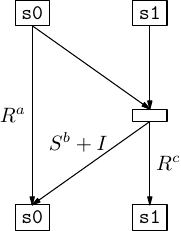}
	\caption{\label{fig:depxoroshiro128}The dependency paths of
		the \xoroshiro[128] linear engine. Data flows from top to bottom: lines
		converging to a box are xor'd together, and labels represent
		$w\times w$ linear transformations applied to the data flowing through the
		line. Note that the linear transformation $S^b+I$ is a xorshift.}
\end{figure}

\subsection{{\fontsize{11pt}{11pt}\selectfont\xoshiro}}

The \xoshiro linear transformation uses only a shift and a rotation. Since
it updates all of the state at each iteration, it is sensible only for
moderate state sizes. We will discuss the
$4w\times 4w$ and $8w\times 8w$ transformations
\[
	\mathscr S_{4w}=\left(\begin{matrix}
			I & I & I   & 0   \\
			I & I & S^a & R^b \\
			0 & I & I   & 0   \\
			I & 0 & 0   & R^b \\
		\end{matrix}\right)\qquad\qquad
	\mathscr S_{8w}=\left(\begin{matrix}
			I & I & I & 0 & 0 & 0 & 0   & 0   \\
			0 & I & 0 & 0 & I & I & S^a & 0   \\
			0 & I & I & 0 & 0 & 0 & 0   & 0   \\
			0 & 0 & 0 & I & 0 & 0 & I   & R^b \\
			0 & 0 & 0 & I & I & 0 & 0   & 0   \\
			0 & 0 & 0 & 0 & I & I & 0   & 0   \\
			I & 0 & 0 & 0 & 0 & 0 & I   & 0   \\
			0 & 0 & 0 & 0 & 0 & 0 & I   & R^b \\
		\end{matrix}\right).
\]

The layout of the matrices above might seem arbitrary, but it is
just derived from the implementation. In Figure~\ref{fig:xoshiro256} and~\ref{fig:xoshiro512}
is it easy to see the algorithmic structure of a
\xoshiro transformation: the second word of the state array is shifted
and stored; then, in order all words of the state array are xor'd with a
different word; finally, the shifted part is xor'd into the next-to-last word of
the state array, and the last word is rotated. The shape of the matrix depends
on the order chosen for the all-words xor sequence.
Figure~\ref{fig:depxoshiro256} shows that also for \xoshiro[256]
dependency paths are very short, and similarly happens for \xoshiro[512].

Note that $\xoshiro$ is not definable for a state of $2w$ bits, and it is too slow for a state of $16w$ bits,
because of the large number of write operations required at each iteration.
\begin{figure}
	%\small
	\begin{verbatim}
    const uint64_t result_plus = s[0] + s[3];
    const uint64_t result_plusplus = rotl(s[0] + s[3], R) + s[0];
    const uint64_t result_starstar = rotl(s[1] * S, R) * T;

    const uint64_t t = s[1] << A;
    s[2] ^= s[0];
    s[3] ^= s[1];
    s[1] ^= s[2];
    s[0] ^= s[3];
    s[2] ^= t;
    s[3] = rotl(s[3], B);
\end{verbatim}
	\caption{\label{fig:xoshiro256}The C code for a \xoshirop[256]/\xoshiropp[256]/\xoshiross[256]
		generator. The state array \texttt{s} contains four $64$-bit unsigned integers, not all zeros.}
\end{figure}

\begin{figure}
	%\small
	\begin{verbatim}
    const uint64_t result_plus = s[0] + s[2];
    const uint64_t result_plusplus = rotl(s[0] + s[2], R) + s[2];
    const uint64_t result_starstar = rotl(s[1] * S, R) * T;

    const uint64_t t = s[1] << A;
    s[2] ^= s[0];
    s[5] ^= s[1];
    s[1] ^= s[2];
    s[7] ^= s[3];
    s[3] ^= s[4];
    s[4] ^= s[5];
    s[0] ^= s[6];
    s[6] ^= s[7];
    s[6] ^= t;
    s[7] = rotl(s[7], B);
\end{verbatim}
	\caption{\label{fig:xoshiro512}The C code for a \xoshirop[512]/\xoshiropp[512]/\xoshiross[512]
		generator. The state array \texttt{s} contains eight $64$bit unsigned
		integers, not all zeros.}
\end{figure}

\begin{figure}[t]
	\centering
	\includegraphics{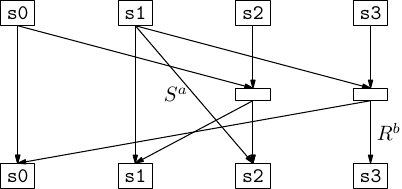}
	\caption{\label{fig:depxoshiro256}The data dependencies of
		the \xoshiro[256] linear engine.}
\end{figure}

\section{Scramblers}
\label{sec:scr}

Scramblers are nonlinear mappings from the state of the linear engine
to a $w$-bit value, which will be the output of the generator. The
purpose of a scrambler is to improve the quality of the raw output of the linear
engine: since in general linear transformations have several useful
provable properties, this is a practical approach to obtain a fast, high-quality
generator.

\subsection{Sum}
\label{sec:plus}

The \texttt{+} scrambler simply adds two words of the state array in $\Z/2^w\Z$.
The choice of words is relevant to the quality of the resulting generator, and we performed several
statistical tests to choose the best pair depending on the underlying engine. The
idea appeared in Saito and Matsumoto's XSadd generator~\cite{SaMXA}, and was
subsequently used by the \xorshiftp family~\cite{VigFSMXG}.

Note that the lowest bit output by the \texttt{+} scrambler is just a xor of bits
following the same linear recurrence, and thus follows, in turn, the same linear
recurrence.
For this reason, we consider \texttt{+} a \emph{weak scrambler}.
As we consider higher bits, there is still a linear recurrence describing their
behavior, but it becomes quickly of such a high linear complexity to become
undetectable. We will discuss in more detail this issue in
Section~\ref{sec:thscr}. In the sample code, the \verb|result_plus|
output is computed using the \texttt{+} scrambler.

\subsection{Multiplication}
\label{sec:star}

The \texttt{*} scrambler multiplies by a constant a chosen word of the state array, and
since we are updating more than a word at a time, the choice of the word is again
relevant. Its only parameter is the multiplier. The multiplier must be odd, so that the scrambling is a bijection; moreover, if
the second-lowest bit set is in position $b$, the lowest $b$ bits of the output
are unmodified, and the following bit is a xor of bit $0$ and bit $b$, so it
follows the same linear
recurrence as the lower bits, as it happens for the lowest bit of the \texttt{+} scrambler. For this
reason, we consider also \texttt{*} a weak scrambler.

We will use multipliers close to $\varphi 2^w$, where $\varphi$ is the golden
ratio, as $\varphi$ is an excellent multiplier for multiplicative
hashing~\cite{KnuACPI}.
To minimize the number of unmodified bits, however, we will adjust the lower bits
in such a way that bit 1 is set. In the sample code, the \verb|result_star|
output is computed using the \texttt{*} scrambler.

\subsection{Sum, rotation, and again sum}
\label{sec:plusplus}

The \texttt{++} scrambler uses two words of the state array: the two words are
summed in $\Z/2^w\Z$, the sum is rotated to the left by $r$ positions, and
finally we add in $\Z/2^w\Z$ the first word to the rotated sum. Note that
the choice \emph{and} the order are relevant---the \texttt{++} scrambler on the first and
last word of state is different from the \texttt{++} scrambler on the last and
first word of state. Besides the choice of words, we have to specify the amount
$r$ of left rotation. Since the rotation moves the highest bits obtained after the first
sum to lower bits, it is easy to set up the parameters so that
there are no bits of low linear complexity in the output. For
this reason, we consider \texttt{++} a \emph{strong scrambler}.
In the sample code, the \verb|result_plusplus|
output is computed using the \texttt{++} scrambler.

\subsection{Multiplication, rotation, and again multiplication}
\label{sec:starstar}

The \texttt{**} scrambler is given by a multiply-rotate-multiply sequence applied to a
chosen word of the state array (again, since we are updating more than
a word at a time, the choice of the word is relevant). It thus requires three
parameters: the first multiplier, the amount of left rotation, and the second
multiplier; both multipliers should be odd, so that the scrambling is a bijection. As in the case of the \texttt{++} scrambler, it is easy to
choose $r$ so that there are no bits of low linear
complexity in the output, so \texttt{**} is a
strong scrambler.

We will mostly use multipliers of the form $2^s+1$, which are usually
computed very quickly, and which have the advantage of being
alternatively implementable with a left shift by $s$ and a sum (the compiler should make
the right choice, but one can also benchmark both implementations).
In the sample code, the \verb|result_starstar|
output is computed using the \texttt{**} scrambler.

\section{Combining linear engines and scramblers}
\label{sec:comb}

In this section we discuss several interesting combinations of linear
engines and scramblers, both for the $64$-bit and the $32$-bit case, and report
results of empirical tests.
We remark that all our generators, being based on linear engines, have
\emph{jump functions} that make it possible to move ahead quickly by any number of
next-state steps. Please refer to~\cite{HMNEJAFRNG,VigFSMXG} for a simple explanation.
% \footnote{We note that
% computation of jump functions can be enormously sped up by exploiting the
% built-in \emph{carry-less multiplication} instructions available on modern
% processors; such instructions, in practice, implement multiplication of
% polynomials on $\Z/2\Z$.}

Part of our experiments use the BigCrush test suite from the well-known
framework TestU01~\cite{LESTU01}. We follow the protocol
described in~\cite{VigEEMXGS}, which we briefly recall. We sample generators by executing BigCrush
starting from several different seeds, using the same setup of~\cite{VigEEMXGS} (in particular,
for $64$-bit generators we generate uniform $32$-bit values by returning first
the lower and then the upper $32$ bits of each output). We consider a test failed if its
$p$-value is outside of the interval $[0.001\..0.999]$.
We call \emph{systematic} a failure that happens for all seeds, and report
systematic failures (a more detailed discussion of this choice can be found
in~\cite{VigEEMXGS}).
Note that we run our tests both on a generator and on
the generator obtained by reversing the order of the 64 bits returned.

Moreover, we ran a new test we designed, aimed at detecting Hamming-weight dependencies~\cite{BlVNTHWD},
that is, dependencies in the number of zeros and ones in each output word,
which are typical of linear generators with sparse transition matrices.
We ran the test until we examined a petabyte ($10^{15}$ bytes) of data, or if we obtained a $p$-value
smaller than $10^{-20}$, in which case we reported the amount of data at which we
stop. The
test is failed by several generators for which the Hamming-weight tests
in TestU01 are unable to find any bias~\cite{BlVNTHWD}, even using several times more data
than for the BigCrush suite.

Not all the generators
we discuss are useful from a practical viewpoint, but discussing several
combinations and their test failures brings to light the limitation of each
component in a clearer way. If the main interest is a practical choice, we
suggest to skip to Section~\ref{sec:choice}.

\subsection{The $64$-bit case}

We consider engines \xoroshiro[128], \xoshiro[256], \xoshiro[512] and
\xoroshiro[1024]; parameters are provided in Table~\ref{tab:par64}.
However, \xoshiro yields generators that have better behavior
with respect to the tests reported in the first eight lines of
Table~\ref{tab:test64}. All linear engines have obvious linear artifacts, but the \xoroshiro engines require an order of magnitude
less data to fail our Hamming-weight dependency test. Note that this is not only
a matter of size, but also of structure: compare $\xoshiro[512]$ and
$\xoroshiro[1024]$. Analogously, the \texttt{+} scrambler deletes all bias detectable with our test from the
$\xoshiro$ generators, but it just improves the resilience of
$\xoroshirop$ by almost three orders of magnitudes.

We then present data on the $\xoroshiro$ generators combined with the \texttt{*}
scrambler: as the reader can notice, the \texttt{*} scrambler does a much better job at
deleting Hamming-weight dependencies, but a worse job at deleting linear
dependencies, as $\xoroshiros[128]$ still fails MatrixRank when reversed. In
Section~\ref{sec:thscr} we will present some theory explaining in detail why
this happens.
Once we switch to the \texttt{++} and \texttt{**} strong scramblers, we are not able to detect any bias.

The parameters for all scramblers are provided in Table~\ref{tab:scr64}. The
actual state words used by the scramblers are described in the code in
Figure~\ref{fig:xoroshiro128}, \ref{fig:xoroshiro1024}, \ref{fig:xoshiro256} and
\ref{fig:xoshiro512}. Note that the choice of word for the $64$-bit engine
\xoroshiro[128] applies also to the analogous $32$-bit engine \xoroshiro[64], and
that the choice for \xoshiro[256] applies also to \xoshiro[128].

% Note that
% for \xoroshiross[1024] we suggest a different strong scrambler, as in this case
% the high sparsity of the matrix defining with the linear engine introduces
% stronger Hamming-weight artifacts.

Our speed tests have been performed on an Intel\textregistered{} Core\texttrademark{} i7-8700B CPU @$3.20$\,GHz
using \text{gcc} 8.3.0.
We used suitable options to keep the compiler from unrolling loops, or extracting
loop invariants, or vectorizing the computation under the hood.

\subsection{The $32$-bit case}

We consider engines $\xoroshiro[64]$ and $\xoshiro[128]$. Most of the
considerations of the previous section are valid, but in this case for
$\xoroshiro$ we suggest \texttt{*} as a weak scrambler: the \texttt{+} scrambler, albeit
faster, in this case is too weak. As in the previous case, the lowest
bits of the generators using a weak scrambler are linear: however, since the
output is just $32$ bits, BigCrush detects this linearity (see failures in the
reverse test).\footnote{We remark that testing subsets of bits of the
	output in the $64$-bit case can lead to analogous results: as long as the subset
	contains the lowest bits in the most significant positions, BigCrush will
	be able to detect their linearity. This happens, for example, if one
	rotates right by one or more positions (or reverses the output) and then tests
	just the upper bits, or if one tests the lowest $32$ bits, reversed.
	Subsets not containing the lowest bits of the generators will exhibit no systematic
	failures of MatrixRank or LinearComp. In principle, the linearity
	artifacts of the lowest bits might be detected also simply by modifying the
	parameters of the TestU01 tests. We will discuss in
	detail the linear complexity of the lowest bits in Section~\ref{sec:thscr}.}

Again, once we switch to the \texttt{**} and \texttt{++} scrambler, we are not able to detect any bias (as for
the \texttt{+} scrambler, we do not suggest to use the \texttt{++} scrambler with
$\xoroshiro[64]$).
The parameters for all scramblers are provided in Table~\ref{tab:scr32}.

\begin{table}
	\caption{\label{tab:test64}Results of tests for $64$-bit generators and three additional popular generators.
		The columns ``S'' and ``R'' report systematic failures in BigCrush
		(MR=MatrixRank, i.e., binary rank; LC=LinearComp, i.e., linear complexity).
		The column ``HWD'' reports the number of bytes generating a $p$-value smaller than $10^{-20}$
		in the test described in~\cite{BlVNTHWD}; no value means that the test was passed after $10^{15}$ bytes.
		The time to emit a 64-bit integer and the number of clock
		cycles per byte (reported by PAPI~\cite{TJYCPDP}) were computed on an
		Intel\textregistered{} Core\texttrademark{} i7-8700B CPU @$3.20$\,GHz.}
	\renewcommand{\arraystretch}{1.1}
	\begin{tabular}{lrrrrr}
		\multirow{2}{*}{Generator}       & \multicolumn{3}{c}{Failures} &
		\multirow{2}{*}{ns/64\,b}        &
		\multirow{2}{*}{cycles/B}                                                                                                             \\
		\cmidrule(rl){2-4}
		                                 & \multicolumn{1}{c}{S}        & \multicolumn{1}{c}{R}   & \multicolumn{1}{c}{HWD} &                 \\
		\hline
		\xoroshiro[128]                  & MR, LC                       & MR, LC                  & $1\times 10^{10}$       & $0.81$ & $0.32$ \\
		\xoshiro[256]                    & MR, LC                       & MR, LC                  & $6\times10^{13}$        & $0.72$ & $0.29$ \\
		\xoshiro[512]                    & MR, LC                       & MR, LC                  & \multicolumn{1}{c}{---} & $0.83$ & $0.39$ \\
		\xoroshiro[1024]                 & MR, LC                       & MR, LC                  & $5\times 10^{12}$       & $1.05$ & $0.42$ \\
		\xoroshirop[128]                 & \multicolumn{1}{c}{---}      & \multicolumn{1}{c}{---} & $5\times10^{12}$        & $0.72$ & $0.29$ \\
		\xoshirop[256]                   & \multicolumn{1}{c}{---}      & \multicolumn{1}{c}{---} & \multicolumn{1}{c}{---} & $0.78$ & $0.31$ \\
		\xoshirop[512]                   & \multicolumn{1}{c}{---}      & \multicolumn{1}{c}{---} & \multicolumn{1}{c}{---} & $0.88$ & $0.35$ \\
		\xoroshirop[1024]                & \multicolumn{1}{c}{---}      & \multicolumn{1}{c}{---} & $4 \times10^{13}$       & $1.05$ & $0.42$ \\
		\xoroshiros[128]                 & \multicolumn{1}{c}{---}      & MR                      & \multicolumn{1}{c}{---} & $0.87$ & $0.37$ \\
		\xoroshiros[1024]                & \multicolumn{1}{c}{---}      & \multicolumn{1}{c}{---} & \multicolumn{1}{c}{---} & $1.11$ & $0.44$ \\
		\xoroshiropp[128]                & \multicolumn{1}{c}{---}      & \multicolumn{1}{c}{---} & \multicolumn{1}{c}{---} & $0.95$ & $0.38$ \\
		\xoshiropp[256]                  & \multicolumn{1}{c}{---}      & \multicolumn{1}{c}{---} & \multicolumn{1}{c}{---} & $0.86$ & $0.34$ \\
		\xoshiropp[512]                  & \multicolumn{1}{c}{---}      & \multicolumn{1}{c}{---} & \multicolumn{1}{c}{---} & $0.99$ & $0.39$ \\
		\xoroshiropp[1024]               & \multicolumn{1}{c}{---}      & \multicolumn{1}{c}{---} & \multicolumn{1}{c}{---} & $1.17$ & $0.47$ \\
		\xoroshiross[128]                & \multicolumn{1}{c}{---}      & \multicolumn{1}{c}{---} & \multicolumn{1}{c}{---} & $0.93$ & $0.42$ \\
		\xoshiross[256]                  & \multicolumn{1}{c}{---}      & \multicolumn{1}{c}{---} & \multicolumn{1}{c}{---} & $0.84$ & $0.33$ \\
		\xoshiross[512]                  & \multicolumn{1}{c}{---}      & \multicolumn{1}{c}{---} & \multicolumn{1}{c}{---} & $0.99$ & $0.39$ \\
		\xoroshiross[1024]               & \multicolumn{1}{c}{---}      & \multicolumn{1}{c}{---} & \multicolumn{1}{c}{---} & $1.17$ & $0.47$ \\
		SplitMix~\cite{SLFFSPNG}         & \multicolumn{1}{c}{---}      & \multicolumn{1}{c}{---} & \multicolumn{1}{c}{---} & $1.14$ & $0.46$ \\
		MT19937-64~\cite{MaNMT,NisT64MT} & LC                           & LC                      & \multicolumn{1}{c}{---} & $2.19$ & $0.94$ \\
		WELL1024a~\cite{PLMILPGBLRM2}    & MR, LC                       & MR , LC                 & \multicolumn{1}{c}{---} & $8.22$ & $3.30$ \\
	\end{tabular}
\end{table}

\begin{table}\caption{\label{tab:par64}Parameters suggested for the $64$-bit
		linear engines used in Table~\ref{tab:test64}. See Section~\ref{sec:poly} for
		an explanation of the ``Weight'' column.}
	\renewcommand{\arraystretch}{1.1}
	\begin{tabular}{lrrrr}
		Engine            & \multicolumn{1}{c}{\texttt{A}} & \multicolumn{1}{c}{\texttt{B}} & \multicolumn{1}{c}{\texttt{C}} & \multicolumn{1}{c}{Weight} \\
		\hline
		\xoroshiro[128]   & 24                             & 16                             & 37                             & 53                         \\
		\xoroshiropp[128] & 49                             & 21                             & 28                             & 63                         \\
		\xoshiro[256]     & 17                             & 45                             & \multicolumn{1}{c}{---}        & 115                        \\
		\xoshiro[512]     & 11                             & 21                             & \multicolumn{1}{c}{---}        & 251                        \\
		\xoroshiro[1024]  & 25                             & 27                             & 36                             & 439                        \\
	\end{tabular}
\end{table}

%\frameit{TODO}{Fix parameters for \protect\xoroshiro[1024]}

\begin{table}\caption{\label{tab:scr64}Parameters suggested for the $64$-bit
		scramblers used in Table~\ref{tab:test64}.
		% Note that for \xoroshiross[1024] we
		% suggest different parameters than for the other linear engines.
	}
	\renewcommand{\arraystretch}{1.1}
	\begin{tabular}{lrrr}
		Scrambler          & \multicolumn{1}{c}{\texttt{S}} & \multicolumn{1}{c}{\texttt{R}} & \multicolumn{1}{c}{\texttt{T}} \\
		\hline
		\texttt{*}         & \texttt{0x9e3779b97f4a7c13}    & \multicolumn{1}{c}{---}        & \multicolumn{1}{c}{---}        \\
		\texttt{**}        & 5                              & 7                              & 9                              \\
		\xoroshiropp[128]  & \multicolumn{1}{c}{---}        & 17                             & \multicolumn{1}{c}{---}        \\
		\xoshiropp[256]    & \multicolumn{1}{c}{---}        & 23                             & \multicolumn{1}{c}{---}        \\
		\xoshiropp[512]    & \multicolumn{1}{c}{---}        & 17                             & \multicolumn{1}{c}{---}        \\
		\xoroshiropp[1024] & \multicolumn{1}{c}{---}        & 23                             & \multicolumn{1}{c}{---}        \\
	\end{tabular}
\end{table}

\begin{table}\caption{\label{tab:test32}Results of tests for $32$-bit
		generators. The column labels are the same as Table~\ref{tab:test64}.}
	\renewcommand{\arraystretch}{1.1}
	\begin{tabular}{lrrr}
		\multirow{2}{*}{Generator} & \multicolumn{3}{c}{Failures}                                                     \\
		\cmidrule(rl){2-4}
		                           & \multicolumn{1}{c}{S}        & \multicolumn{1}{c}{R}   & \multicolumn{1}{c}{HWD} \\
		\hline
		\xoroshiro[64]             & MR, LC                       & MR, LC                  & $5\times 10^8$          \\
		\xoshiro[128]              & MR, LC                       & MR, LC                  & $3.5\times 10^{13}$     \\
		\xoroshiros[64]            & \multicolumn{1}{c}{---}      & MR, LC                  & \multicolumn{1}{c}{---} \\
		\xoshirop[128]             & \multicolumn{1}{c}{---}      & MR, LC                  & \multicolumn{1}{c}{---} \\
		\xoshiropp[128]            & \multicolumn{1}{c}{---}      & \multicolumn{1}{c}{---} & \multicolumn{1}{c}{---} \\
		\xoroshiross[64]           & \multicolumn{1}{c}{---}      & \multicolumn{1}{c}{---} & \multicolumn{1}{c}{---} \\
		\xoshiross[128]            & \multicolumn{1}{c}{---}      & \multicolumn{1}{c}{---} & \multicolumn{1}{c}{---} \\
	\end{tabular}
\end{table}

\begin{table}\caption{\label{tab:par32}Parameters suggested for the $32$-bit linear
		engines used in Table~\ref{tab:test32}.}
	\renewcommand{\arraystretch}{1.1}
	\begin{tabular}{lrrrr}
		Engine         & \multicolumn{1}{c}{\texttt{A}} & \multicolumn{1}{c}{\texttt{B}}
		               & \multicolumn{1}{c}{\texttt{C}} & \multicolumn{1}{c}{Weight}                                    \\
		\hline
		\xoroshiro[64] & 26                             & 9                              & 13                      & 31 \\
		\xoshiro[128]  & 9                              & 11                             & \multicolumn{1}{c}{---} & 55 \\
	\end{tabular}
\end{table}

\begin{table}\caption{\label{tab:scr32}Parameters suggested for the $32$-bit
		scramblers used in Table~\ref{tab:test32}.}
	\renewcommand{\arraystretch}{1.1}
	\begin{tabular}{lrrr}
		Generator        & \multicolumn{1}{c}{\texttt{S}} & \multicolumn{1}{c}{\texttt{R}}
		                 & \multicolumn{1}{c}{\texttt{T}}                                                            \\
		\hline
		\xoroshiros[64]  & \texttt{0x9E3779BB}            & \multicolumn{1}{c}{---}        & \multicolumn{1}{c}{---} \\
		\xoroshiross[64] & \texttt{0x9E3779BB}            & 5                              & 5                       \\
		\xoshiropp[128]  & \multicolumn{1}{c}{---}        & 7                              & \multicolumn{1}{c}{---} \\
		\xoshiross[128]  & 5                              & 7                              & 9                       \\
	\end{tabular}
\end{table}

\subsection{Choosing a generator}
\label{sec:choice}

Our $64$-bit proposals for an all-purpose generator are
$\xoshiropp[256]$ and $\xoshiross[256]$. Both sport excellent speed, a state
space that is large enough for any parallel application,\footnote{With $256$ bits
	of state, $2^{64}$ sequences of length $2^{64}$ starting at $2^{64}$ random points in the state
	space have an overlap probability of less than $2^{-64}$, which is entirely
	negligible~\cite{NauEBP,VigPORSPNG}. One can also use jumping to guarantee the absence of overlap.} and pass all tests we are aware of.
In theory, $\xoshiropp[256]$ uses simpler operations and can be easily parallelized using Intel's extended AVX2 instruction set;
however, it also accesses two words of state.  Moreover, even if the \texttt{**}
scrambler in $\xoshiross[256]$ is specified using multiplications, it can be
implemented using only a few shifts, xors, and sums. Another difference is that
$\xoshiross[256]$ is $4$-dimensionally equidistributed (see
Section~\ref{sec:ed}), whereas $\xoshiropp[256]$ is just $3$-dimensionally
equidistributed, albeit this difference will not have any effect in practice. On
the other hand, as we will see in Section~\ref{sec:thscr}, the bits of $\xoshiropp[256]$ have
higher linear complexity.

If, however, one has to generate only $64$-bit floating-point numbers (by
extracting the upper $53$ bits), or if the mild linear artifacts in its lowest
bits are not considered problematic, $\xoshirop[256]$ is a faster generator with analogous
statistical properties.\footnote{On our hardware,
	generating a floating-point number with $53$ significant bits takes $1.15$\,ns. This datum can be compared,
	for example, with the \texttt{dSFMT}~\cite{MuMPSDPFTNUAT}, which
	using extended SSE2 instructions provides a double with $52$ significant bits only
	in $0.90$\,ns, but fails linearity tests and our Hamming-weight dependency test~\cite{BlVNTHWD}.}

There are however some cases in which $256$ bits of state are considered too
much, for instance when throwing a very large number of lightweight threads, or
in embedded hardware. In this case, a similar discussion applies to
$\xoroshiropp[128]$, $\xoroshiross[128]$, and $\xoroshirop[128]$, with the
\textit{caveat} that the latter has mild problems with our Hamming-weight
dependency test: however, bias can be detected only after $5$\,TB of data, which
makes it unlikely to affect applications in any way.

Finally, there might be cases that we cannot foresee in which more bits of state
are necessary:
$\xoshiropp[512]$, $\xoshiross[512]$, and $\xoshirop[512]$ should be the first choice, switching to
$\xoroshiropp[1024]$, $\xoroshiross[1024]$, or $\xoroshiros[1024]$ if even more bits are necessary. In
particular, if rotations are available $\xoroshiros[1024]$ is an obvious better
replacement for $\xorshifts[1024]$~\cite{VigEEMXGS}. As previously discussed,
however, it is very difficult to motivate from a theoretical viewpoint a
generator with more than $256$ bits of state.\footnote{We remark that, as
	discussed in Section~\ref{sec:poly}, it is possible to create $\xoroshiro$
	generators with even more bits of state.}

Turning to $32$-bit generators, $\xoshiropp[128]$, $\xoshiross[128]$, and $\xoshirop[128]$ have a
role corresponding to $\xoshiropp[256]$, $\xoshiross[256]$, and $\xoshirop[256]$ in the $64$-bit
case: $\xoshiropp[128]$ and $\xoshiross[128]$ are our first choice, while $\xoshirop[128]$ is our choice
for $32$-bit floating-point generation.
For $\xoroshiro[64]$ we suggest however a \texttt{*} scrambler, as the \texttt{+} scrambler
turns out to be too weak for this simple engine.

The state of a generator should be in principle seeded with truly random
bits. If only a 64-bit seed is available, we suggest using a
\textsc{SplitMix}~\cite{SLFFSPNG} generator, initialized with the given seed,
to fill the state array of our generators, as research has shown that
initialization must be performed with a generator radically different in nature from
the one initialized to avoid correlation on similar seeds~\cite{MWKCDIPNG}.\footnote{It is immediate to define a $32$-bit version of \textsc{SplitMix} to initialize
	$32$-bit generators.} Since
\textsc{SplitMix} is an equidistributed generator, the resulting initialized state
will never be the all-zero state. Notice, however, that using a $64$-bit seed
only a minuscule fraction of the possible initial states will be obtainable.
In any case, the seed must be stored for repeatability.

\section{Polynomials and full period}
\label{sec:poly}
One of the fundamental tools in the investigation of linear transformations is
the \emph{characteristic polynomial}. If $M$ is the $n\times n$ matrix representing the
transformation associated with a linear engine the characteristic polynomial is
\[
	p(x)=\det(M-xI).
\]
The associated linear engine has \emph{full period} (i.e., maximum-length period $2^n-1$) if and only if $p(x)$
is \emph{primitive} over $\Z/2\Z$~\cite{LiNIFFA}, that is, if $p(x)$
is irreducible and if $x$ has maximum period in the ring of polynomials over
$\Z/2\Z$ modulo $p(x)$. By enumerating all possible parameter choices and checking primitivity
of the associated polynomials we can discover all full-period
linear engines.

In particular, every bit of a linear engine satisfies a linear recurrence
with characteristic polynomial $p(x)$. Different bits emit
different outputs because they return sequences from different
starting points in the orbit of the recurrence.

The \emph{weight} of $p(x)$ is the number of terms in $p(x)$, that is,
the number of nonzero coefficients. It is considered a good property for a linear
engine of this kind\footnote{Technically, the criterion applies to the linear recurrence represented by the characteristic polynomial.
	The behavior of a linear engine, however, depends also on the relationships among all its state bits, so the degree criterion
	must always be weighed against other evidence.} that the weight is close to
half the degree, that is, that the polynomial is neither too sparse nor too
dense~\cite{ComHCRBS}.

\subsection{Word polynomials}

A matrix $M$ of size $kw\times kw$ can be viewed as a $k\times k$ matrix on the
ring of $w\times w$ matrices. At that point, some generalization of the
determinant to noncommutative rings can be used to obtain a characteristic
polynomial $p_w(x)$ for $M$ (which will be a polynomial with $w\times w$
matrices as coefficients):
if the determinant of $p_w(x)$ on the base ring (in our case, $\Z/2\Z$) is equal to the characteristic
polynomial of $M$, then $p_w(x)$ is a \emph{word polynomial of size $w$ for
	$M$}.

The main purpose of word polynomials is to make easier the computation of the
characteristic polynomial of $M$ (and thus of the determinant of $M$), in particular for large matrices.
Characteristic polynomials can be computed also by applying the Berlekamp--Massey algorithm to a bit of
the the linear engine: in our experience, on large matrices the word-polynomial approach, if
applicable, is faster. The difference in speed is not relevant, however, as the primitivity check
is by far the most expensive step.

In all our examples $w$ is the intended output size of the linear engine, but in
some cases it might be necessary to use a smaller block size, say, $w/2$, to
satisfy commutativity conditions: one might speak, in that case, of the \emph{semi-word
	polynomial}.
For blocks of size one, the word polynomial is simply the characteristic
polynomial; the commutation constraints are trivially satisfied.

If all $w\times w$ blocks of $M$ commute as elements of the ring of $w\times w$ matrices, a very well-known result from
Bourbaki~\cite{BouAEM} shows that computing the determinant of $M$ in the
commutative ring $R$ of $w\times w$ matrices generated by the $w\times w$
blocks of $M$ one has
\begin{equation}
	\label{eq:det}
	\operatorname{det}(\Det(M))=\operatorname{det}(M),
\end{equation}
where ``det'' denotes the determinant in the base ring (in our case, $\Z/2\Z$)
whereas ``Det'' denotes the determinant in $R$. This equivalence provides a very handy
way to compute easily the determinants of large matrices with a block
structure containing several zero blocks and commuting non-zero blocks: one
simply operates on the blocks of the matrix as if they were scalars.

However, if $M$ is the matrix associated with a linear engine,
$\Det(M)$ can be used also to characterize how the
current state of the linear engine depends on its previous states. Indeed, since
we are working in a commutative ring the Cayley--Hamilton theorem holds, and
thus $M$ is a root of its characteristic polynomial: if we let
$p_w(x)=\Det(M-xI)$, then $p_w(M)=0$. In more detail, if
\[p_w(x)=x^k+A_{k-1}x^{k-1}+\dots +A_1x + A_0, \]
where the $A_i$'s are $w\times w$ matrices of the ring $R$ generated by the blocks of $M$, then
\[M^k+M^{k-1}A_{k-1}+\dots +MA_1 + A_0=0.\]
Note that in the formula above we are multiplying $k\times k$
matrices on $R$ by scalar coefficients in $R$ (i.e., as usual a scalar coefficient represents
a diagonal matrix containing the scalar along the diagonal).

Thus, given a sequence of states
$\bm s_0$,~$\bm s_1=\bm s_0  M$,~$\bm s_2=\bm s_0 M^2$, $\dots\,$,~$\bm
	s_k = \bm s_0 M^k$ we have
\begin{equation}
	\label{eq:mrmm}
	\bm s_k = \bm
	s_{k-1}A_{k-1}+\dots + \bm s_1A_1 + \bm s_0A_0.
\end{equation} This recurrence
makes it possible to compute the next state of a linear engine knowing its previous $k$ states.
Note that in the equation above in practice we are multiplying each of the $k$ blocks
of length $w$ of the $\bm s_i$'s by the $A_i$'s.

This consideration may seem trivial, as we already know how to compute the next
state given the previous state---a multiplication by $M$ is sufficient---but the
recurrence is true for every $w$-bit block of the state array. Said
otherwise, no matter which word of the state array we choose as output, we can
predict the next output using the equation above knowing the last $k$ outputs (i.e., the previous $k$ states of the chosen word).
Another way of looking at the same statement is that the word polynomial expresses
the linear engine described by $M$ using Niederreiter's \emph{multiple-recursive
	matrix method}~\cite{NieMRMMPNG}, much like the characteristic polynomial
expresses a single output bit as a linear recurrence.

Recurrence~(\ref{eq:mrmm}) will work not only in the commutative case, but
also whenever a sufficiently powerful
extension of the Cayley--Hamilton has been proved for the class of matrices
under examination (e.g., see Theorem 14 of~\cite{CFRAPMM1}, which can
be used to prove~(\ref{eq:mrmm}) for \xorshift linear engines with multiple-word state~\cite{VigEEMXGS}).

\subsection{The noncommutative case}

The observations of the previous section cannot help us in computing the characteristic polynomials
of \xoroshiro or \xoshiro, because their matrices contain non-commuting blocks.
There are two issues in generalizing the arguments we made about the
commutative case: first, we need a notion of noncommutative determinant; second,
we need to know whether~(\ref{eq:det}) generalizes to our case.

Both issues are addressed by recent results by
Sothanaphan~\cite{SotDBMNB}. One starts by defining a (standard) notion of
determinant for noncommutative
rings by fixing the order of the products in Leibniz's formula. In particular,
we denote with $\Detr$ the
\emph{row-determinant} of an $n\times n$ matrix $M$ on a
noncommutative base ring:
\begin{equation}
	\label{eq:detr}
	\Detr(M)=\sum_{\pi\in S_n}\operatorname{sgn}(\pi)
	M_{0,\pi(0)}M_{1,\pi(1)}\cdots M_{n-1,\pi(n-1)}
\end{equation}
Note that the definition is based on Leibniz's formula, but the order
of the products has been fixed. Then, Theorem~1.2 of~\cite{SotDBMNB} shows that
\begin{equation}
	\label{eq:detrdet}
	\det(\Detr(M)) = \det(M),
\end{equation}
provided that blocks in different columns, but not in the first row,
commute. In other words, one can compute the characteristic polynomial
of $M$ by first computing the ``row'' characteristic polynomial of $M$ by
blocks and then computing the determinant of the resulting matrix. By the
definition we gave, in this case $\Detr(M-xI)$ is a word polynomial for $M$.

The row-determinant is (trivially) antisymmetric with respect to the
permutation of columns.\footnote{In our case, that is, on the base field
	$\Z/2\Z$ there is no difference between ``symmetric'' and ``antisymmetric'' as sign
	change is the identity. For the sake of generality, however, we will recall the
	properties we need in the general case.} Moreover, a basic
property of row-determinants depends on a commutativity
condition on the matrix entries (blocks): if $M$ has \emph{weakly
	column-symmetric commutators}, that is, if
\[
	M_{ij}M_{kl}-M_{kl}M_{ij} = M_{il}M_{kj}-M_{kj}M_{il} \quad \text{whenever
		$i\neq k$ and $j\neq l$,}
\]
then the row-determinant is antisymmetric with respect to the permutation of
rows~\cite{CSSND}.

Dually, we can define the \emph{column-determinant}
\begin{equation}
	\label{eq:detc}
	\Detc(M)=\sum_{\pi\in S_n}\operatorname{sgn}(\pi)
	M_{\pi(0),0}M_{\pi(1),1}\cdots M_{\pi(n-1),n-1}.
\end{equation}
All recalled properties can be easily dualized to the case of the
column-determinant:
in particular,
\begin{equation}
	\label{eq:detcdet}
	\det(\Detc(M)) = \det(M),
\end{equation}
provided that blocks in different rows, but not in the first column,
commute; if $M$ has \emph{weakly
	row-symmetric commutators}, that is, if
\[
	M_{ij}M_{kl}-M_{kl}M_{ij} = M_{kj}M_{il}-M_{il}M_{kj} \quad \text{whenever
		$i\neq k$ and $j\neq l$,}
\]
then the column-determinant is antisymmetric with respect to the permutation of
columns~\cite{CSSND}.

Finally, if $M$ is \emph{weakly commutative}~\cite{CSSND}, that is, $M_{ij}$ and
$M_{kl}$ commute whenever $i\neq k$ and $j\neq l$ (i.e., noncommuting blocks lie either
on the same row or on the same column) the two determinants are the same, as all
products in~(\ref{eq:detr}) and~(\ref{eq:detc}) can be rearranged arbitrarily.\footnote{We remark that if the only aim is to compute easily
	the characteristic polynomial, one can rearrange columns and rows at will until~(\ref{eq:detrdet}) or~(\ref{eq:detcdet}) is true, because these operations
	cannot change the value of the determinant on $\Z/2\Z$.}

\subsubsection{\xoroshiro} Our first try is to check the conditions
for~(\ref{eq:detrdet}) on the transition matrix $\mathscr X_{kw}$, but there is
no easy way to modify the $\mathscr X_{kw}$ to satisfy them. However, it is easy
to check that $\mathscr X_{kw}$ has weakly row-symmetric commutators, so we can
move its next-to-last column to the first one, and then the resulting matrix
falls into the conditions for~(\ref{eq:detcdet}). We thus obtain a word
polynomial based on the column-determinant:
\begin{align*}
	\Detc & \bigl(\mathscr X_{kw}-xI\bigr)                                         \\
	      & =\Detc\left(\begin{matrix}
			                    R^a + S^b + I & xI     & 0      & \cdots & 0      & R^c    \\
			                    0             & I      & xI     & \cdots & 0      & 0      \\
			                    0             & 0      & I      & \cdots & 0      & 0      \\
			                    \cdots        & \cdots & \cdots & \cdots & \cdots & \cdots \\
			                    xI            & 0      & 0      & \cdots & I      & 0      \\
			                    S^b + I       & 0      & 0      & \cdots & 0      & R^c+xI
		                    \end{matrix}\right) \\
	      & =\bigl(R^a + S^b + I \bigr) \bigl(R^c+xI\bigr) +x \bigl(R^c+xI\bigr)
	x^{k-2} +\bigl(S^b + I\bigr)R^c                                                \\
	      & =x^k I +
	x^{k-1}R^c + x\bigl(R^a + S^b + I\bigr) + R^{a+c}.
\end{align*}
Since $\mathscr X_{kw}$ is noncommutative (and it does not satisfy currently known
extensions of the Cayley--Hamilton theorem), it is unlikely that the polynomial above can express
the linear transformation as in~(\ref{eq:mrmm}): and indeed it cannot.
This lack of commutativity, however, does not
hamper our ability to use the word polynomial to compute the characteristic
polynomial: simply, we cannot obtain directly a recurrence like~(\ref{eq:mrmm}).

Nonetheless, we can check empirically whether \emph{some} of the bits of the output are
predictable using the word polynomial (i.e., whether they satisfy the linear constraints it
expresses). Empirically $4/5$ of the bits of each word of state can be predicted
using the polynomial above when $k=2$ (\xoroshiro[128]), and the ratio becomes
about $1/2$ for $k=16$ (\xoroshiro[1024]).\footnote{The empirical observation
	about predicted bits are based on the parameters of Table~\ref{tab:par64}
	and~\ref{tab:par32}: different parameters will generate different results.}

\subsubsection{\xoshiro}

In this case, we have to perform an \textit{ad hoc} maneuver to move $\mathscr
	S_{4w}$ and $\mathscr S_{8w}$ into a form amenable to the computation of a
word polynomial by row-determinant: we have to exchange the
first two rows. It is very easy to see that this operation cannot modify the
row-determinant because every element of the first row commutes with every
element of the second row: thus, in the products of~(\ref{eq:detr}) the first
two elements can always be swapped.

At that point, by (a quite tedious) Laplace expansion along the first row we get
\[
	\Detr \bigl(\mathscr S_{4w}-xI\bigr)
	=x^4I+x^3\bigl(R^b+I\bigr)+x^2\bigl(S^a+R^b\bigr)+x\bigl(S^a+I\bigr)\bigl(R^b+I\bigr)+\bigl(S^a+I\bigr)R^b
\]
and
\begin{multline*}
	\Detr \bigl(\mathscr S_{8w}-xI\bigr)
	\\=x^8I+x^7\bigl(R^b+I\bigr)+x^6\bigl(R^b+I\bigr)+x^5\bigl(S^a+R^b+I\bigr)+x^4\bigl(S^a+I\bigr)\bigl(R^b+I\bigr)\\+x^3\bigl(S^aR^b+R^b+S^a\bigr)
	+x^2\bigl(S^a+I\bigr)\bigl(R^b+I\bigr)+x\bigl(S^aR^b+R^b+I\bigr)+R^b.
\end{multline*}
In this case, we have sometimes a behavior similar to the commutative case: for $\mathscr S_{4w}$
(\xoshiro[256]), the second word of state can be predicted
exactly\footnote{Incidentally, if we reverse multiplication order in the coefficients, the first word can be
	predicted exactly instead.}; for the other words, about two
thirds of the bits can be predicted.
In the case of $\mathscr S_{8w}$ (\xoshiro[512]), all words except the
last one can be predicted exactly; for the last one, again about two thirds of the bits can be
predicted.

\begin{table}\caption{\label{tab:xorpol}Number of \xoroshiro primitive
		polynomials depending on word size and state size.}
	\renewcommand{\arraystretch}{1.1}
	\begin{tabular}{lc|rrrrrrr}
		                           &                            & \multicolumn{7}{c}{State size in bits}                                                                   \\
		                           &                            & \multicolumn{1}{c}{$64$}               & \multicolumn{1}{c}{$128$}
		                           & \multicolumn{1}{c}{$256$}  & \multicolumn{1}{c}{$512$}              &
		\multicolumn{1}{c}{$1024$} & \multicolumn{1}{c}{$2048$} & \multicolumn{1}{c}{$4096$}                                                                               \\
		\hline
		                           & $16$                       & $26$                                   & $21$                      & $7$   & $3$   & $1$   & $0$  & $0$  \\
		$w$                        & $32$                       & $250$                                  & $149$                     & $59$  & $41$  & $16$  & $5$  & $6$  \\
		                           & $64$                       & $$                                     & $1000$                    & $491$ & $261$ & $129$ & $59$ & $37$ \\
	\end{tabular}
\end{table}

\begin{table}\caption{\label{tab:xordeg}Maximum weight of a \xoroshiro primitive
		polynomial depending on word size and state size.}
	\renewcommand{\arraystretch}{1.1}
	\begin{tabular}{lc|rrrrrrr}
		                          &                           & \multicolumn{7}{c}{State size in bits}                                                                                           \\
		                          &                           & \multicolumn{1}{c}{$64$}               & \multicolumn{1}{c}{$128$}  &
		\multicolumn{1}{c}{$256$} & \multicolumn{1}{c}{$512$} & \multicolumn{1}{c}{$1024$}             & \multicolumn{1}{c}{$2048$} & \multicolumn{1}{c}{$4096$}                                 \\
		\hline
		                          & $16$                      & $37$                                   & $45$                       & $73$                       & $35$  & $41$  & $$    & $$    \\
		$w$                       & $32$                      & $39$                                   & $67$                       & $115$                      & $201$ & $187$ & $195$ & $143$ \\
		                          & $64$                      & $$                                     & $75$                       & $139$                      & $263$ & $475$ & $869$ & $1303$ \\
	\end{tabular}
\end{table}

\begin{table}\caption{\label{tab:xospol}Number of \xoshiro primitive polynomials
		depending on word size and state size.}
	\renewcommand{\arraystretch}{1.1}
	\begin{tabular}{lc|rrrr}
		                          &                           & \multicolumn{4}{c}{State size in
		bits}                                                                                                                            \\
		                          &                           & \multicolumn{1}{c}{$64$}         & \multicolumn{1}{c}{$128$} &
		\multicolumn{1}{c}{$256$} & \multicolumn{1}{c}{$512$}                                                                            \\
		\hline
		                          & $16$                      & $1$                              & $0$                       & $$  & $$  \\
		$w$                       & $32$                      & $$                               & $1$                       & $0$ & $$  \\
		                          & $64$                      & $$                               & $$                        & $4$ & $4$ \\
	\end{tabular}
\end{table}

\begin{table}\caption{\label{tab:xosdeg}Maximum weight of a \xoshiro primitive
		polynomial depending on word size and state size.}
	\renewcommand{\arraystretch}{1.1}
	\begin{tabular}{lc|rrrr}
		                          &                           & \multicolumn{4}{c}{State size in
		bits}                                                                                                                                \\
		                          &                           & \multicolumn{1}{c}{$64$}         & \multicolumn{1}{c}{$128$} &
		\multicolumn{1}{c}{$256$} & \multicolumn{1}{c}{$512$}                                                                                \\
		\hline
		                          & $16$                      & $33$                             & $$                        & $$    & $$    \\
		$w$                       & $32$                      & $$                               & $55$                      & $$    & $$    \\
		                          & $64$                      & $$                               & $$                        & $131$ & $251$ \\
	\end{tabular}
\end{table}

\subsection{Full-period linear engines}

Using the word polynomials just described we computed exhaustively all
parameters providing primitive characteristic polynomials, and thus full-period
linear engines, using the Fermat algebra system~\cite{LewF}, stopping the search
at $4096$ bits of state.\footnote{The reason why the number 4096 is
	relevant here is that we know the factorization of Fermat's numbers $2^{2^k}+1$
	only up to $k=11$. When more Fermat numbers will be factorized, it will be
	possible to find linear engines with a larger state
	space.}

Table~\ref{tab:xorpol} and~\ref{tab:xospol}
report the number of primitive polynomials, whereas Table~\ref{tab:xordeg} and~\ref{tab:xosdeg} report the
maximum weight of a primitive polynomial. As one can expect, we find that there are many more
full-period \xoroshiro instances at many
more different state sizes than \xoshiro, due to the additional parameter. We note that
by Proposition~7.1 from~\cite{VigEEMXGS} all full-period linear engines have the property
that each output bit has full period, too.

We did not discuss $16$-bit generators, but there is a \xoshiro and several
\xoroshiro choices available.

\section{Equidistribution}
\label{sec:ed}

\emph{Equidistribution} is a uniformity property of pseudorandom number
generators: a generator with $kw$ bits of state and $w$ output bits is
$d$-dimensionally equidistributed if when we consider the vector of the first
$d$ output values over all possible states of the generator, each vector
appears the same number of times~\cite{LecMECTG}.
In practice, in linear generators over the whole output every $d$-tuple of
consecutive output values must appear $2^{w(k-d)}$ times for $d\leq k$, except for the zero $d$-tuple, which appears
$2^{w(k-d)}-1$ times, as the zero state is not used.
In particular, $1$-dimensionally equidistributed generators with $w$ bits of
state emit each $w$-bit value exactly one time, except for the value zero, which
is never emitted.\footnote{A more refined definition might consider only a subset of bits,
	in which case equidistribution in larger dimensions is possible~\cite{LEPFRNG}.}
We will start by discussing the equidistribution of our linear engines (without
scramblers).

Typically, linear generators (e.g., \xorshift) update cyclically a position of
their state array. In this case, the simple fact that the generator has full
period guarantees that the generator is equidistributed in the maximum
dimension, that is, $k$. However, since our linear engines update more than one
position at a time, full period is not sufficient, and different words
of the state may display different equidistribution properties.

Testing for equidistribution in the maximum dimension for a word of the state array is easy using
standard techniques, given the update matrix $\mathscr M$ of the linear
engine:
if we consider the $j$-th word, the square matrix obtained juxtaposing the
$j$-th block columns of $\mathscr M^0=I$,~$\mathscr M^1=\mathscr M$, $\mathscr M^2$,
$\dots\,$,~$\mathscr M^{k-1}$ must be invertible (the inverted matrix returns,
given an output vector of $k$ words, the state that will emit the vector).
It is straightforward to check that every word of $\xoroshiro$ (for every state size) and of \xoshiro[512] is equidistributed in the maximum dimension.
The words of \xoshiro[256] are equidistributed in the maximum dimension
except for the third word, for which equidistribution depends on the
parameters; we will not use it.

The \texttt{*} and \texttt{**} scramblers cannot alter the equidistribution of the full output of a
linear engine as they just remap sequences bijectively (however, note that if we
start to consider the equidistribution of a subset of output bits
this is no longer true).
Thus, all our generators using such scramblers are $k$-dimensionally
equidistributed (i.e., in the maximum dimension).

We are left with proving equidistribution results for our generators based on the
\texttt{+} scrambler and on the \texttt{++} scrambler. For $d$-dimensionally
equidistributed linear engines that update cyclically a single position, adding
two consecutive outputs can be easily proven to provide a $(d-1)$-dimensionally
equidistributed generator. However, as we already noticed our linear engines
update more than one position at a time: we thus proceed to develop a general technique,
which can be seen as an extension of the standard technique to prove equidistribution of a purely
linear generator, and will be used in the following sections.

Note that since we have to mix operations from two algebraic structures, throughout this section the
symbols $+$ and $-$ will denote operations in $\Z/2^w\Z$, whereas $\oplus$ will
denote sum in $(\Z/2\Z)^w$.

\subsection{A general technique for equidistribution of \texttt{+}/\texttt{++}-scrambled generators}

For a linear engine with $k$ words of state, we consider a vector of variables
$\bm x = \la x_0,x_1,\dots,x_{k-1}\ra$ representing the state of the engine.
Then, for each $0\leq i <d$, where $d$ is the target equidistribution, we add
variables $t_i, u_i$ and equations $t_i=\bigl(\bm x \mathscr M^i\bigr)_p$,
$u_i=\bigl(\bm x \mathscr M^i\bigl)_q$, where $p$ and $q$ are the two state
words to be used by the \texttt{+} or \texttt{++} scrambler.

Given a target vector of output values $\la v_0, v_1, \dots, v_{d-1}\ra$, we
would like to show that there are $2^{w(k-d)}$ possible values of $\bm x$ that
will give the target vector as output.
This condition can be expressed by equations on the $t_i$'s and the $u_i$'s
involving the arithmetic of $\Z/2^w\Z$. In the case of the \texttt{+} scrambler,
we have equations $v_i=t_i+u_i$; in particular, $p$ and $q$ can be exchanged
without affecting the equations.
In the case of the \texttt{++} scrambler, instead, if $t_i$ denotes the first
word of state used by the scrambler (see Section~\ref{sec:plusplus}) we have
$v_i = (u_i+t_i)R^r+ t_i$, but there is no way to derive $t_i$ from $u_i$; the
dual statement is true if $t_i$ denotes the second word of state used by the
scrambler.

Now, to avoid mixing operations in $\Z/2^w\Z$ and $(\Z/2\Z)^w$ we will first
solve using standard linear algebra the $x_i$'s in terms of the $t_i$'s and
the $u_i$'s. At that point, we will be handling a new set of constraints in
$(\Z/2\Z)^w$ containing only $t_i$'s and $u_i$'s:
using a limited amount of \textit{ad hoc} reasoning, we will have to show how by
choosing $k-d$ parameters freely we can satisfy at the same time both the new
set of constraints and the equations on $\Z/2^w\Z$ induced on the $t_i$'s and
the $u_i$'s by the choice of a scrambler. If we will be able to do so, we will
have parameterized all occurrences of $\la v_0, v_1, \dots, v_{d-1}\ra$ in the
output using $k-d$ parameters, so such occurrences must be at most $2^{w(k-d)}$.
But since there are $2^{wk}$ $d$-dimensional output vectors (including the
all-zero output associated with the all-zero state), by pigeonholing the occurrences must be
exactly $2^{w(k-d)}$.

\subsection{\xoroshiro}

\begin{proposition}
	\label{prop:xoroshirop}
	A \xoroshirop generator with $w$
	bits of output and $kw$ bits of state applying the \texttt{+} scrambler (see Section~\ref{sec:plus}) to the first and last word of
	state is $(k-1)$-dimensionally equidistributed.
\end{proposition}
\begin{proof}
	For the case $k=2$ the full period of the underlying \xoroshiro generator
	proves the statement. If $k>2$, denoting with the $t_i$'s the first word and
	with the $u_i$'s the last word our technique applied to $\mathscr X_{kw}$ provides
	equations
	\begin{align*}
		\label{eq:xoroshirop}
		t_i     & = x_i                & 0 & \leq i \leq k - 2 \\
		u_0     & = x_{k-1}                                    \\
		u_{i+1} & = (t_i\oplus u_i)R^c & 0 & \leq i \leq k - 3
	\end{align*}
	Thus, the only constraint on the $t_i$'s and $u_i$'s is the last equation.
	It is immediate that once we assign a value to $u_0$
	we can derive a value for $t_0 = v_0-u_0$, then a value for $u_1$ and so on.
\end{proof}

Note that the claim of Proposition~\ref{prop:xoroshirop} cannot be extended to $k$-dimensional
equidistribution.
Consider the full-period $5$-bit generator with $10$ bits of state and
parameters $a=1$, $b=3$, and $c=1$. As a \xoroshiro generator it is
$2$-dimensionally equidistributed, but it is easy to verify that the sequence of
outputs of the associated \xoroshirop generator is not $2$-dimensionally
equidistributed (it is, of course, $1$-dimensionally equidistributed by
Proposition~\ref{prop:xoroshirop}).

The proof of Proposition~\ref{prop:xoroshirop} can be easily extended to the
case of a \xoroshiropp generator.
\begin{proposition}
	\label{prop:xoroshiropp}
	A \xoroshiropp generator with $w$
	bits of output and $kw$ bits of state applying the \texttt{++} scrambler (see Section~\ref{sec:plusplus})
	to the last and first word of
	state is $(k-1)$-dimensionally equidistributed.
	If $k=2$, also  applying the \texttt{++} scrambler the first and last word of state yields a
	$1$-dimensionally equidistributed generator.
\end{proposition}
\begin{proof}
	We use the same notation as in Proposition~\ref{prop:xoroshirop}.
	For the case $k=2$ the full period of the underlying \xoroshiro generator
	proves the statement, as there is no constraint, so the only equation between
	$t_0$ and $u_0$ is either
	$u_0 = (v_0 - t_0)R^{-r}- t_0$, if we are scrambling the first and the last
	words of state, or $t_0 = (v_0 - u_0)R^{-r}- u_0$, if we are scrambling
	the last and the first one.

	Otherwise, we proceed as in the proof of
	Proposition~\ref{prop:xoroshirop}. Since we are scrambling the last and first
	word, we have $t_i = (v_i - u_i)R^{-r}- u_i$, and the proof can be completed in
	the same way.
\end{proof}

The counterexample we just used for \xoroshirop can be used in the \xoroshiropp
case, too, to show that the claim of Proposition~\ref{prop:xoroshiropp}
cannot be extended to $k$-dimensional
equidistribution. Moreover, the full-period \xoroshiropp $4$-bit generator with
$16$ bits of state and parameters $a=3$, $b=1$ and $c=2$ is not even $3$-dimensionally
equidistributed if we scramble the first and last word (instead of the last and
the first), showing that the stronger statement for $k=2$ does not extend to
larger values of $k$.

\subsection{\xoshiro}

\begin{proposition}
	\label{prop:xoshiro256p}
	A \xoshirop generator with $w$ bits of output and $4w$ bits of state
	applying the \texttt{+} scrambler (see Section~\ref{sec:plus}) to the first and last word of state is $3$-dimensionally equidistributed.
\end{proposition}
\begin{proof}
	In this case, denoting with the $t_i$'s the first word and
	with the $u_i$'s the last word, our technique applied to $\mathscr S_{4w}$
	provides the constraints
	\[
		t_i = t_{i+1} \oplus  u_{i+1} R^{-b} \qquad 0\leq i \leq 1.
	\]
	But if we choose $t_2$ arbitrarily, we can immediately compute
	$u_2=v_2-t_2$, then $t_1$, then $u_1=v_1-t_1$, and finally $t_0$.
	% Consider a triple of values $\la v_0, v_1, v_2\ra$ and let us compute the first
	% and last words of the two states following $\la v_0-z,x,y,z\ra$, whose
	% output is of course $v_0$. We obtain the following constraints:
	% \begin{align*}
	% v_1 -  \bigl(xR^a+
	% zR^a\bigr) &= (v_0-z)+ x + z  \\
	%  v_2 - \bigl((v_0 -
	%  z)R^a+ x\bigl(R^{2a}+ R^a\bigr)+ yR^a+ zR^{2a}\bigr)  &=  xR^a+ y+ z(I+ R^a) 
	% \end{align*}
	% We now add a variable $u=x+ z$, which gives us
	% \begin{align*}
	% u &= x +  z\\
	% z &= v_0 - ((v_1-uR^a)+ u )\\  
	% v_2 -  \bigl((v_0 - z)R^a+
	% u\bigl(R^{2a}+ R^a\bigr)+ zR^a +  yR^a\bigr) &= uR^a+ y+ z 
	% \end{align*}
	% Replacing the value of $z$ into $(v_0-z)$ in the last equation leads to some
	% simplification:
	% \begin{align*}
	% u &= x +  z\\
	% z &= v_0 - ((v_1-uR^a)+ u )\\  
	% v_2 -  \bigl((v_1-uR^a)R^a+
	% uR^{2a}+ zR^a +  yR^a\bigr) &= uR^a+ y+ z 
	% \end{align*}
	% Finally, we add a variable $t=uR^a+ z+ y$:
	% \begin{align*}
	% t&=uR^a+ z+ y \\
	% u &= x +  z\\
	% z &= v_0 - ((v_1-uR^a)+ u )\\  
	% t &= v_2 - \bigl((v_1-uR^a)+ t\bigr)R^a
	% \end{align*}
	% We are now in the position of making the dependence of $u$ on $t$ explicit in
	% the last equation, as $u$ appears just once:
	% \[
	% u = \bigl(v_1-\bigl((v_2-t)R^{-a}+  t\bigr)\bigr)R^{-a}.   
	% \]
	% If we choose a value for $t$ we obtain a value for $u$ and a value
	% for $z$ by the last two equations, and then the first two equations give the values
	% of $x$ and $y$. Thus, a triple $\la v_0, v_1, v_2\ra$ can appear at most $2^w$ times, which implies
	% that all triples appear exactly $2^w$ times.
\end{proof}

Once again, the claim of Proposition~\ref{prop:xoshiro256p} cannot be extended to $4$-dimensional
equidistribution. The only possible $2$-bit \xoshiro generator with $8$ bits of
state has full period but it is easy to verify that the associated \xoshirop
generator is not $4$-dimensionally equidistributed.

Now, we prove an analogous equidistribution result for \xoshirop[512].
\begin{proposition}
	\label{prop:xoshiro512p}
	A \xoshirop generator with $w$ bits of output and $8w$ bits of state and
	applying the \texttt{+} scrambler (see Section~\ref{sec:plus}) to the first and third word of state is $7$-dimensionally
	equidistributed.
\end{proposition}
\begin{proof}
	Denoting with the $t_i$'s the first word and
	with the $u_i$'s the third word and
	applying again our general technique, we obtain the constraints
	\[
		t_i = u_i \oplus  u_{i+1} \qquad 0\leq i \leq 5.\\
	\]
	Choosing a value for $u_0$ (and thus $t_0=v_0-u_0$) gives by the first equation
	$u_1=t_0\oplus u_0$ and thus $t_1=v_1 - u_1$, by the second equation
	$u_2=t_1\oplus u_1$ and $t_2=v_2 - u_2$, and so on.
\end{proof}

Note that the claim of Proposition~\ref{prop:xoshiro512p} cannot be extended to $8$-dimensional
equidistribution: the \xoshirop generator associated with the only full period
$5$-bit \xoshiro generator with $40$ bits of state ($a=2$, $b=3$) is not
$8$-dimensionally equidistributed.

\begin{proposition}
	\label{prop:xoshiro256pp}
	A \xoshiropp generator with
	$w$ bits of output and $4w$ bits of state and scrambling the first and last
	words of state is $3$-dimensionally equidistributed.
\end{proposition}
\begin{proof}
	The proof uses the same notation of
	Proposition~\ref{prop:xoshiro256p}, and proceeds in the same way: the equations
	we obtain are the same, and due to the choice of scrambler we have $u_i = (v_i - t_i)R^{-r}- t_i$, so can
	obtain the $u_i$'s from the $t_i$'s.
	% 
	%  Consider a triple of values $\la v_0, v_1,
	% v_2\ra$ and a state $\la x, y, z, w\ra$. Denoting with $R^r$ the permutation applied to the sum of the
	% state words, the condition that the first three outputs of the \xoshiropp
	% generator are $v_0$, $v_1$ and $v_2$ gives three contraints:
	% \begin{align*}
	% (v_0 - w)R^{-r} - w &= y \\
	% (v_1 - (yR^a + wR^a))R^{-r} - (yR^a + wR^a) &= x + y + z\\
	% \bigl(v_2 - \bigl(xR^a + y\bigl(R^a+ R^{2a}\bigr)
	% + zR^a + wR^{2a}\bigr)\bigr)R^{-r} - \bigl(xR^a + y\bigl(R^a+
	% R^{2a}\bigr) + zR^a + wR^{2a}\bigr) &= x+ yS^b+ w 
	% \end{align*}
	% We now notice that we can solve the second constraint for $x + z$:
	% \[ x+z = \bigl((v_1 - (yR^a + wR^a))R^{-r} - (yR^a + wR^a)\bigr) + y \]
	% Finally, we can rearrange the last constraint so to make explicit $x$ in terms
	% of $y$, $w$ and $x+z$:
	% \[
	% x = \Bigl(\bigl(v_2- \bigl((x+ z)R^a +  y\bigl(R^a+ R^{2a}\bigr)
	% +  wR^{2a}\bigr)\bigr)R^{-r} -\bigl((x+ z)R^a + 
	% y\bigl(R^a+ R^{2a}\bigr) +  wR^{2a}\bigr)\Bigr) + yS^B+ 
	% w
	% \]
	% Plugging in the value for $x+ z$ we obtain $x$ in terms of $y$ and $w$.
	% Thus, for every $w$ there is a unique $y$, a unique $x$ and ultimately a unique
	% $z$ satisfying the constraints; in particular, a triple $\la v_0, v_1, v_2\ra$
	% can appear at most $2^w$ times, which implies that all triples appear exactly $2^w$ times.
\end{proof}

\begin{proposition}
	\label{prop:xoshiro512pp}
	A \xoshiropp generator with
	$w$ bits of output and $8w$ bits of state scrambling the third and first words
	of state is $7$-dimensionally equidistributed.
\end{proposition}
\begin{proof}
	The proof uses the same notation of
	Proposition~\ref{prop:xoshiro512p}, and proceeds in the same way: the equations
	we obtain are the same, and due to the choice of scrambler we have $t_i = (v_i - u_i)R^{-r}- u_i$, so we can obtain the $t_i$'s from the $u_i$'s.
\end{proof}

The counterexamples for Proposition~\ref{prop:xoshiro256p} and~\ref{prop:xoshiro512p} used to prove that their
claims cannot be extended to higher-dimensional equidistribution
work also for Proposition~\ref{prop:xoshiro256pp} and~\ref{prop:xoshiro512pp}, respectively.

\section{Escaping zeroland}

We show in Figure~\ref{fig:ez} the speed at which the generators hitherto
examined ``escape from zeroland''~\cite{PLMILPGBLRM2}:
linear engines need some time to get from an initial state with a small number
of bit set to one to a state in which the ones are approximately
half (famously, the Mersenne Twister requires millions of
iterations), and while scrambling reduces this phenomenon, it is nonetheless
detectable.
The figure shows a measure of escape time given by the ratio of
ones in a window of 4 consecutive 64-bit values sliding over the first 1000
generated values, averaged over all possible seeds with exactly one bit set (see~\cite{PLMILPGBLRM2} for a detailed
description).

\begin{figure}
	\centering
	\includegraphics[scale=.8]{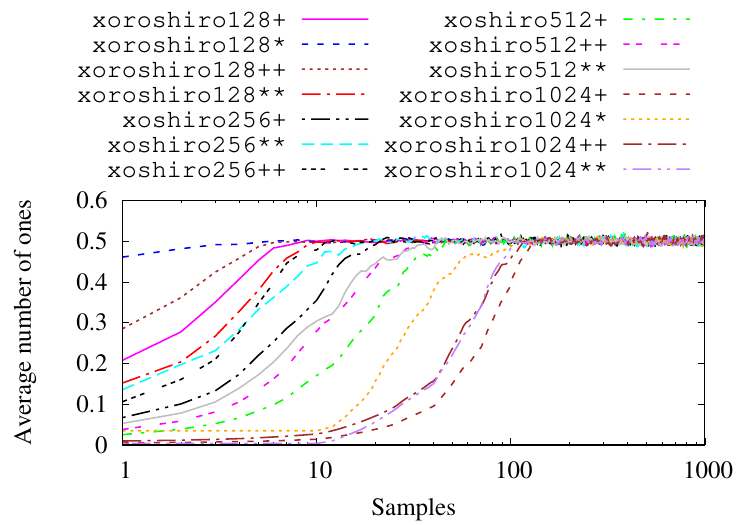}
	\caption{\label{fig:ez}Convergence to ``half of the bits are ones in average'' plot.}
\end{figure}

\section{A theoretical analysis of scramblers}
\label{sec:thscr}

We conclude the paper by discussing our scramblers from a theoretical point of
view. We cast our discussion in the same theoretical framework as that of
\emph{filtered linear-feedback shift registers}. A filtered LFSR
is given by an underlying LFSR and by a Boolean function that
is applied to the state of the register. The final output is the output of the Boolean function.
If the LFSR updates one bit at a time, we can see the Boolean function as sliding on the
sequence of bits generated by the LFSR, emitting a scrambled output. The purpose
of filtering a LFSR is that of making it more difficult to guess its next bit: the analogy
with linear engines and scramblers is evident, as every scrambler can be seen as a set of $w$
boolean functions applied to the state of the linear engine.
There are however a few differences:
\begin{itemize}
	\item we use only primitive polynomials;
	\item  we use several Boolean functions, and we are concerned
	      with the behavior of their combined outputs;
	\item  we do not apply a Boolean function to a sliding
	      window of the same LFSR: rather, we have $kw$ copies of the same LFSR whose
	      state is different, and we apply our set of Boolean functions to their single-bit output concatenated;
	\item we are not free to design our favorite Boolean functions: we are restricted to the ones
	      computable with few arithmetic and logical operations;
	\item  we are not concerned
	      with predictability in the cryptographic sense, but just in the elimination of linear artifacts,
	      that is, failures in tests for binary rank, linear complexity, and Hamming-weight
	      dependencies.
\end{itemize}

We will see that many basic techniques coming from the cryptographic analysis
of filtered LFSRs can be put to good use in our case.
We will bring along a very simple example:
a full-period \xorshift linear engine with $w=3$ and $6$ bits of state~\cite{VigEEMXGS}. Its parameters are
$a=1$ (left shift), $b=2$ (right shift), $c=1$ (right shift), and its characteristic polynomial is $p(x)=x^6 + x^5 + x^3 + x^2
	+ 1$.

\subsection{Representation by generating functions}

We know that all bits of a linear engine satisfy linear recurrences with the same
characteristic polynomial, but we can be more precise: we can fix a nonzero
initial state and compute for each bit the \emph{generating function} associated
with the bit (see~\cite{KleSC} for a detailed algorithm).
Such functions have at the denominator the \emph{reciprocal polynomial} $x^n
	p(1/x)$ ($n$ here is the degree of $p$), whereas the numerator (a polynomial
of degree less than $n$) represents the initial state.  In our example,
representing the first word using $x_0$ (lowest bit), $x_1$, $x_2$, the second word using $y_0$, $y_1$ and $y_2$, and using
as initial state all bits set to zero except for $x_0$, we have
\begin{align*}
	F_{x_0}(z) & = \frac{z^5 + z^2 + z + 1}{z^6 + z^4 + z^3 + z + 1} & F_{y_0}(z) & = \frac{z^5 + z^4 + z^3 + z^2 + z}{z^6 + z^4 + z^3 + z + 1} \\
	F_{x_1}(z) & = \frac{z^5 + z^2}{z^6 + z^4 + z^3 + z + 1}         & F_{y_1}(z) & = \frac{z^4 + z}{z^6 + z^4 + z^3 + z + 1}                   \\
	F_{x_2}(z) & = \frac{z^5 + z^4}{z^6 + z^4 + z^3 + z + 1}         & F_{y_2}(z) & = \frac{z^4 + z^3}{z^6 + z^4 + z^3 + z + 1}                 \\
\end{align*}
If we write the formal series associated with each function, the $i$-th coefficient will give exactly the $i$-th output
of the corresponding bit of the linear engine.

The interest in the representation by generating function lies in the fact that now we can perform
some operations on the bits. For example, to study the lower bits of a generator
using the $+$ scrambler to our linear engine, we add two bits, and we can easily compute the associated function, as adding coefficients
is the same as adding functions:
\[
	F_{x_0+y_0}(z) = F_{x_0}(z)+F_{y_0}(z) \\= \frac{z^5 + z^2 + z + 1}{z^6 + z^4 + z^3 + z + 1}+
	\frac{z^5 + z^4 + z^3 + z^2 + z}{z^6 + z^4 + z^3 + z + 1} = \frac{z^4 + z^3 + 1}{z^6 + z^4 + z^3 + z + 1}.
\]
However, we are now stuck, because
addition over $\Z/2^w\Z$ needs more than just xors. With
\lst xw and \lst yw representing the
bits, from least significant to most significant, of two $w$-bit words,
to represent their arithmetic sum over $\Z/2^w\Z$ we can define the result bits
$s_i$ and the carry bits $c_i$ using the recurrence
\begin{align}
	\label{eq:recsum}
	s_i & = x_i + y_i + c_{i-1}        \\
	c_i & = (x_i+ y_i)c_{i-1} + x_iy_i
\end{align}
where $c_{-1} = 0$. This recurrence is fundamental because carries are the only source of nonlinearity
in our scramblers (even multiplication by a constant can be turned into a series of shifts and sums).
It is clear that to continue to the higher bits we need to be able
to multiply two sequences, but multiplying generating functions, unfortunately, corresponds to a convolution of coefficients.

\subsection{Representation in the splitting field}

We now start to use the fact that the characteristic polynomial of our linear engine is primitive.
Let $\mathbf E$ be the \emph{splitting field} of a primitive polynomial $p(-)$ of degree
$n$ over $\Z/2\Z$~\cite{LiNIFFA}. In particular, $\mathbf E$ can be
represented as $(\Z/2\Z)[\alpha]/p(\alpha)$, that is, by polynomials in $\alpha$ computed modulo $p(\alpha)$, and in
that case by primitivity the zeroes of $p(-)$ are exactly the powers
\[
	\alpha, \alpha^2, \alpha^4, \alpha^8, \ldots, \alpha^{2^{n-1}},
\]
that is, the powers having exponents in the \emph{cyclotomic coset}
$C=\bigl\{\,1,2,4,8,\ldots, 2^{n-1}\,\bigr\}$. Note that
$\alpha^{2^n}=\alpha$ in $\mathbf E$.
Every rational function $f(z)$ representing the output of a bit of the linear engine can then be
expressed as a sum of partial fractions
\begin{equation}
	\label{eq:parfrac}
	f(z) = \sum_{c\in C}\frac{\beta_c}{1-z\alpha^c},
\end{equation}
where $\beta_i\in \mathbf E$, $\beta_c\neq 0$. As a consequence~\cite{KleSC},
the $j$-th bit
$b_j$ of the sequence associated with $f(z)$ has an explicit description:
\begin{equation}
	\label{eq:repr}
	b_j =\sum _{c\in C}\beta_c\bigl(\alpha^c\bigr)^j.
\end{equation}
This property makes it possible to compute the sum
of two sequences and the (output-by-output) product of two sequences. We just
need to compute the sum or the product of the representation~(\ref{eq:repr}).
The sum of two sequences is just a term-by-term sum, whereas
in the case of a product we obtain a convolution.
In both cases, we might experience \emph{cancellation}---some of the $\beta$'s
might become zero. But, whichever operation we apply, we will obtain in the end
for a suitable set $S\subseteq[2^n]$ a representation of the form
\begin{equation}
	\label{eq:repr2}
	\sum _{c \in S}\beta_c\alpha^c.
\end{equation}
with $\beta_c\neq 0$. The cardinality of $S$ is now exactly the degree of the
polynomial at the denominator of the rational function
\[
	g(z) = \sum_{c\in S}\frac{\beta_c}{1-z\alpha^c}
\]
associated with the new sequence, that is, its \emph{linear
	complexity}~\cite{KleSC}.
In our example, the coefficients of the representation~(\ref{eq:repr}) of $x_0$ are
\begin{align*}
	\beta_1 & = \alpha^4 + \alpha^3                    & \beta_8    & =\alpha^4 + \alpha^3 + \alpha^2 + \alpha  \\
	\beta_2 & =\alpha^5 + \alpha^3 + \alpha^2 + \alpha & \beta_{16} & = \alpha^5 + \alpha^4 + \alpha^3 + \alpha \\
	\beta_4 & =\alpha^5 + \alpha^4 + \alpha^2 + 1      & \beta_{32} & = \alpha^5 + \alpha^2 + \alpha            \\
\end{align*}
and similar descriptions are available for the other bits, so
we are finally in the position of computing exactly the values of the
recurrence~(\ref{eq:recsum}): we simply have to use the representation
in the splitting field to obtain a representation of $s_i$, and then revert to
functional form using~(\ref{eq:parfrac}).

\begin{figure}
	\tiny
	\begin{align*}
		F_{s_0}(z) & =  \frac{z^4 + z^3 + 1}{z^6 + z^4 + z^3 + z + 1}                                                  \\
		F_{s_1}(z) & = \frac{z^{13} + z^9 + z^8 + z}{z^{15} + z^{14} + z^{11} + z^7 + z^4 + z^3 + 1}                   \\
		F_{s_2}(z) & = \frac{z^{37} + z^{35} + z^{32} + z^{30} + z^{27} + z^{25} + z^{24} + z^{20} + z^{19} + z^{17} +
			z^{14} + z^{11} + z^8 + z^2}{z^{41} + z^{39} + z^{34} + z^{32} + z^{30} + z^{28} + z^{27}
			+ z^{26} + z^{24} + z^{23} + z^{17} + z^{16} + z^{15} + z^{14} + z^{13} + z^{11} + z^9 +
			z^8 + z^7 + z^6 + z^5 + z^3 + 1}
	\end{align*}
	\caption{\label{fig:sumgf}The generating functions of the three bits of the \xorshiftp generator.}
\end{figure}

The result is shown in Figure~\ref{fig:sumgf}: as it is easy to see, we can
still express the bits of \xorshiftp as LFSRs, but their linear complexity
rises quickly (remember that every generator with $n$ bits of state is a linear generator of
degree $2^n$ with characteristic polynomial $x^{2^n}+1$, so ``linear'' should
always mean ``linear of low degree'').

Note that the generating function is irrelevant for our purposes: the only relevant fact is
that the representation in the splitting field of the first bit has $6$ coefficients, that
of the second bit $15$ and that of the third bit $41$, because, as we have
already observed, these numbers are equal to the linear complexity of
the bits of our \xorshiftp generator.
% 
% Thus, to compute products one just goes through the representation in the
% splitting field, and then reverts to functional form using~(\ref{eq:parfrac}).
% This approach makes it possible to compute the generating functions of all bits
% emitted by the $+$ scrambler.
Unfortunately, this approach can be applied only to state arrays of less than a
dozen bits: as the linear complexity increases due to the influence of carries,
the number of terms in the representation~(\ref{eq:repr2}) grows quickly, up to
being unmanageable. Thus, this approach is limited to the analysis of small
examples or the construction of counterexamples.

\subsection{Representing scramblers by polynomials}
\label{sec:raprscr}

A less exact but more practical approach to the analysis of the scrambled output
of a generator is that of studying the scrambler in isolation. To do so, we are going
to follow the practice of the theory of filtered LFSRs: we will represent the scramblers
as a sum of \emph{Zhegalkin polynomials}, that is, \emph{squarefree polynomials} over $\Z/2\Z$.
Due to the peculiarity of the field, no coefficients or exponents are necessary. If we can
describe the function as a sum of distinct polynomials, we will say that
the function is in \emph{algebraic normal form} (ANF).  For example, the $3$-bit scrambler of our
example generator can be described by expanding recurrence~(\ref{eq:recsum}) into
the following three functions in ANF:
\begin{align*}
	S_0(x_0,x_1,x_2,y_0,y_1,y_2) & = x_0 + y_0                                  \\
	S_1(x_0,x_1,x_2,y_0,y_1,y_2) & = x_1 + y_1 + x_0y_0                         \\
	S_2(x_0,x_1,x_2,y_0,y_1,y_2) & = x_2 + y_2 + x_1y_1 + x_0y_0x_1 + x_0y_0y_1 \\
\end{align*}

There is indeed a connection between the \emph{polynomial degree} of a Boolean function, that is,
the maximum degree of a polynomial in its ANF, the linear complexity of
the bits of a linear engine, and the linear complexity of the bit returned
by the Boolean function applied to the engine state. We will use the standard notation $[n]=\{\?0,1,2,\ldots,n-1\?\}$.

\begin{lemma}
	\label{lemma:subset}
	Let $\mathbf E$ be the splitting field of a primitive polynomial $p(x)$ of
	degree $n$, represented by polynomials in $\alpha$ computed modulo $p(\alpha)$.
	Then, there is a tuple $\langle t_0, t_1, \ldots t_{k-1}\rangle\in [n]^k$
	such that
	\[
		\prod_{i\in [k]}\alpha^{2^{t_i}}=\alpha^{c}
	\]
	iff there is an $S\subseteq [n]$ with $0<|S|\leq k$ and
	\[
		c= \sum_{s\in S}2^s.
	\]
\end{lemma}

\begin{proof}
	First we show that the all $c$'s are of the form above.
	When all the $t_i$'s are distinct, we have trivially $S=\{\?t_i \mid 0\leq
		i<k\?\}$. If $t_i=t_j$
	\[\alpha^{2^{t_i}}\alpha^{2^{t_j}}= \alpha^{2\cdot
			2^{t_i}}=\alpha^{2^{t_i+1}},\]
	remembering that computations of exponents of $\alpha$ are to be made modulo
	$2^n-1$. Thus, the problem is now reduced to a smaller tuple, and we can
	argue by induction that the result will be true of some $S\subseteq
		[k-1]\subseteq[k]$.

	Now we show that for every $S$ as in the statement there exists a corresponding
	tuple. If $|S|=k$, this is obvious. Otherwise, let $|S|=j$ and \lst s{j} be an
	enumeration of the elements of $S$. Then, the $k$-tuple
	\[
		s_0, s_1, \ldots , s_{j-2}, s_{j-1} - 1, s_{j-1} - 2, \ldots,
		s_{j-1} - k + j + 1, s_{j-1} - k + j, s_{j-1} - k + j,
	\]
	where the operations above are modulo $n$, gives rise exactly to the
	set $S$, as
	\[
		\alpha^{2^{s_{j-1} - 1}}\alpha^{2^{s_{j-1} - 2}} \cdots \alpha^{2^{s_{j-1} - k + j + 1}}\alpha^{2^{ s_{j-1} - k + j}}\alpha^{2^{ s_{j-1} - k + j}} = \alpha^{2^{s_{j-1}}}.
	\]
\end{proof}

An immediate consequence of the previous lemma is that there is a bound on the increase
of linear complexity that a Boolean function, and thus a scrambler, can
induce:
\begin{proposition}
	If $f$ is a Boolean function of $n$ variables with polynomial degree $d$ in ANF
	and $x_i$, $0\leq i< n$, are the bits of a
	linear engine with $n$ bits of state, then the rational function representing
	$f\bigl(x_0, x_1,\ldots , x_{n-1})$ has linear complexity
	at most
	\begin{equation}
		\label{eq:ub}
		U(n,d) = \sum_{j=1}^d {n\choose j}.
	\end{equation}
\end{proposition}

The result is obvious, as the number of possible nonzero coefficients of the splitting-field
representation of $f\bigl(x_0,x_1\dots,x_{n-1}\bigr)$ is bounded by $U(n,d)$ by
Lemma~\ref{lemma:subset}. Indeed, $U(n,d)$ is well known:
it is the standard bound on the linear complexity of a filtered LFSR. Our case is different, as we are
applying Boolean functions to bits coming from different instances of the same LFSR, but the mathematics is the same.

There is also another inherent limitation: a uniform scrambler
on $m$ bits cannot have polynomial degree $m$:
\begin{proposition}
	\label{prop:wminus1}
	Consider a vector of
	$n$ Boolean functions on $m$ variables such that the preimage
	of each vector of $n$ bits contains exactly $2^{m-n}$ vectors of $m$ bits.
	Then, no function in the vector can have the (only) monomial of degree
	$m$ in its ANF.
\end{proposition}
\begin{proof}
	Since the vector of functions maps the same number of input values to each
	output value, if we look at each bit and consider its value over all possible vectors
	of $m$ bits, it must be zero $2^{m-1}$ times, one $2^{m-1}$
	times. But all monomials of degree less than $m$ evaluate to an even number of
	zeroes and ones. The only monomial of degree $m$ evaluates to one exactly once.
	Hence, it cannot appear in any of the polynomial functions.
\end{proof}

Getting back to our example, the bounds for linear complexity of the bits of our
\xorshiftp generator are ${6\choose 1}= 6$, ${6\choose 1}+ {6\choose 2}=21$,
and ${6\choose 1}+ {6\choose 2}+ {6\choose 3}=41$.
From Figure~\ref{fig:sumgf}, the first and last bits attain the upper bound~(\ref{eq:ub}), whereas the
intermediate bit does not.
However, Lemma~\ref{lemma:subset} implies that every subset of $S$ might
be associated with a nonzero coefficient. If this does not happen, as in the case of the intermediate bit,
it must be the case that all the contributions for that subset of $S$ canceled
out.

The amount of cancellation happening for a specific combination of linear engine
and scrambler can in principle be computed exactly using the splitting-field representation, but
as we have discussed this approach does not lend itself to computations beyond very small
generators. However, we gathered some empirical evidence by computing the
polynomial degree of the Boolean function associated with a bit
using~(\ref{eq:recsum}) and then by
measuring directly the linear complexity using the Berlekamp--Massey
algorithm~\cite{KleSC}: with careful implementation, this technique can be applied much beyond where the splitting-field representation can get. The algorithm
needs an upper bound on the linear complexity to return a reliable result, but
we have~(\ref{eq:ub}).
We ran extensive tests on several generators, the largest being $12$-bit
generators with $24$ bits of state.
The results are quite uniform: unless the state array of the linear engine is
tiny, if the characteristic polynomial is primitive, cancellation is an extremely rare event.

These empirical finds suggest that it is a good idea to independently study
scramblers as Boolean functions, and in particular estimating or computing their
polynomial degree. Then, given a class of generator, one should gather some
empirical evidence that cancellation is rare, and at that point use the
upper bound~(\ref{eq:ub}) as an estimate of linear complexity. This is the
approach that we will follow in the following sections.

We remark however that a high polynomial degree is not sufficient to
guarantee to pass all tests related to linearity. The problem is that such tests
depend on the joint output of the Boolean functions we are considering.
Moreover, there is a great difference between having a high polynomial degree and
passing a linear-complexity or binary-rank test.

For example, consider the following pathological Boolean function that could be
part of a scrambler:
\begin{equation}
	\label{eq:pathological}
	x_{w-1} + \prod_{i\in[w-1]} x_i.
\end{equation}
This function has very high polynomial degree, and thus a likely high
linear complexity.
The problem is that if, say, $w=64$ from a practical viewpoint it is
indistinguishable from $x_{w-1}$, as the ``correction'' that raises its linear
complexity rarely happens. If the state array is small, this bit will fail
all linearity tests.
A single high-degree monomial is not sufficient in
isolation, despite Lemma~\ref{lemma:subset}, so we will look
for scramblers represented by a large number of monomials.

As a last counterexample, and cautionary tale, we consider the scrambler given
by a change of sign, that is, multiplication by the all-ones word. It is trivial
to write this scrambler using negated variables, but when we expand it in ANF
we get
\begin{equation}
	\label{eq:neg}
	\bar x_{w-1} + \prod_{k\in [w-1]} \bar x_k = 1+x_{w-1} + \prod_{k\in [w-1]}
	\bigl(1+ x_k) = 1+x_{w-1} + \prod_{S\subseteq[w-1]}\prod_{k\in S} x_k.
\end{equation}
In other words, the ANF contains all monomials
formed with all other bits, but the Boolean function is still as pathological
as~(\ref{eq:pathological}), as there is no technical difference between $x_i$
and $\bar x_i$. Too few monomials are problematic, but too many are, too.

\subsection{The \texttt{+} scrambler}
\label{ref:plusscr}
We conclude this part of the paper with a detailed discussion of each scrambler,
using their representations by squarefree polynomials, as discussed in the previous section.
We start from the \texttt{+} scrambler, introduced in Section~\ref{sec:plus}.
Recurrence~(\ref{eq:recsum}) can be easily unfolded to a closed form for the scrambled bit $s_b$:
\begin{multline}
	\label{eq:sumasn}
	s_b = x_b + y_b + \sum_{i=1}^{b} x_{i-1} y_{i-1}
	\sum_{S\subseteq[b-i]}\prod_{j\in S}x_{i+j}\prod_{j\in [b-i]\setminus S}y_{i+j}\\
	=x_b + y_b + \sum_{i=1}^{b} x_{i-1} y_{i-1}
	\prod_{j\in[b-i]}\bigl(x_{i+j}+y_{i+j}\bigr).
\end{multline}
If the $x_i$'s and the $y_i$'s are distinct, the expressions above are in ANF:
there are exactly $2^b+1$ monomials with maximum degree $b+1$. Thus, if the
underlying linear engine has $n$ bits of state the linear-degree bound for bit
$b$ will be $U(n,b+1)$, where $U(-,-)$ is defined by~(\ref{eq:ub}).

An important observation is that no monomial appears in two instances of the
formula for different values of $b$. This implies that any linear combination of
bits output by the \texttt{+} scrambler has the same linear complexity as the
bit of highest degree, and at least as many monomials: we say in this case that there
is no \emph{polynomial degree loss}.
Thus, except for the very lowest bits, we expect that no linearity
will be detectable.

In Table~\ref{tab:plusdeg} we report, using~(\ref{eq:ub}), the estimated linear complexity of the lowest bits of
some generators. The lowest values have
also been verified using the Berlekamp--Massey algorithm: as expected, we could
not detect any linear-degree loss; running the algorithm on the largest values is unfeasible.
While an accurate linear-complexity test might catch the fourth lowest bit
of \xoroshirop[128], the degree raises quickly to the point the linearity is
undetectable.

\begin{table}
	\renewcommand{\arraystretch}{1.1}
	\centering
	\begin{tabular}{rrrrrr}
		{\small\xoroshirop[128]} & {\small\xoshirop[256]} & {\small\xoshirop[512]} & {\small\xoroshirop[1024]} & {\small\xoroshirop[64]} & {\small\xoshirop[128]} \\
		\hline
		128                      & 256                    & 512                    & 1024                      & 64                      & 128                    \\
		8256                     & 32896                  & 131328                 & 524800                    & 2080                    & 8256                   \\
		349632                   & 2796416                & 22370048               & 178957824                 & 43744                   & 349632                 \\
		11017632                 & 177589056              & 2852247168             & 45723987200               & 679120                  & 11017632               \\
		275584032                & 8987138112             & 290367762560           & 9336909979904             & 8303632                 & 275584032              \\
	\end{tabular}
	\caption{\label{tab:plusdeg}Estimated linear complexity of the five lowest bits
		of generators (the first line is bit 0) using the \texttt{+} scrambler.}
\end{table}

The situation for
Hamming-weight dependencies is not so good, however, as empirically
(Table~\ref{tab:test64}) we have already observed that \xoroshiro engines
still fail our test (albeit using three orders of magnitude more data).
We believe that this is due to the excessively regular structure of the
monomials.

Note that if the underlying linear engine is $d$-dimensionally equidistributed,
the scrambler generator will be in general at most $(d-1)$-dimensionally
equidistributed (see Section~\ref{sec:ed}).

\subsection{The \texttt{*} scrambler}
\label{ref:starscr}

We now discuss the \texttt{*} scrambler, introduced in
Section~\ref{sec:star}, in the case of a multiplicative constant of
the form $2^s+1$. This case is particularly interesting  because it is very fast
on recent hardware; in particular, $(2^s+1)\cdot x = x + (x \ll s)$, where the
sum is in $\Z/2^w\Z$, which provides a multiplication-free implementation.
Moreover, as we will see, the analysis of the $2^s+1$ case sheds light on the general case, too.

Let $z=(2^s+1)x$. Specializing~(\ref{eq:sumasn}), we have that
$z_b = x_b$ when $b<s$; otherwise, $b = c + s\geq s$ and
\begin{multline}
	\label{eq:star}
	z_b = z_{c+s} = x_{c+s} + x_c + \sum_{i=1}^c x_{i-1+s}\, x_{i-1}
	\sum_{S\subseteq[c-i]}\prod_{j\in S}x_{i+j+s}\prod_{j\in [c-i]\setminus
		S}x_{i+j}\\
	= x_{c+s} + x_c + \sum_{i=1}^c x_{i-1+s} \,x_{i-1}\prod_{k\in
		[c-i]}\bigl(x_{i+k+s} + x_{i+k}\bigr).
	%=x_{c+s} + x_c + \sum_{i=1}^c x_{i-1+s} x_{i-1}
	%\sum_{S\subseteq[c-i]}\prod_{j\in S}x_{i+j+s}\prod_{j\in [c-i]\setminus
	%(S\cup S+s)}x_{i+j}.\\
\end{multline}
However, contrarily to~(\ref{eq:sumasn}) the expressions above do not denote an
ANF, as the same variable may appear many times in the same
monomial.

We note that the monomial $x_sx_0x_{s+1}\cdots x_{s+c-1}$,
which is of degree $c+1$, appears only and always in the function associated with $y_b$, $b>s$. Thus, bits
with $b\leq s$ have degree one, whereas bits $b$ with $b>s$ have degree $b-s+1$.
In particular, as in the case of \texttt{+}, there is no polynomial degree loss when combining different bits.

In the case of a generic (odd) constant $m$, one has to modify
recurrence~(\ref{eq:recsum}) to start including shifted bits at the right
stage, which creates a very complex monomial structure.
Note, however, that bits after the second-lowest bit set in $m$ cannot modify the
polynomial degree. Thus, the decrease of Hamming-weight dependencies we
observe in Table~\ref{tab:test64} even for \xoroshiros is not due to a higher polynomial
degree with respect to \texttt{+} (indeed, the opposite is true), but to a
richer structure of the monomials. The degree reported for the \texttt{+}
scrambler in Table~\ref{tab:plusdeg} can indeed be adapted to the present case: one has
just to copy the first line as many times as the index of
the second-lowest bit set in $m$.

% There is, however, a marked difference in the number and structure of
% monomials of the Boolean functions associated with each bit as $s$ varies.
% The first observation is that the case $s=1$ is peculiar as we obtain
% \begin{multline*}
% x_{c+1} + x_c + \sum_{i=1}^c x_i \,x_{i-1}\prod_{k\in
% [c-i]}\bigl(x_{i+k+1} + x_{i+k}\bigr)
% \\=x_{c+1} + x_c(1+ x_{c-1}) + \sum_{i=1}^{c-1} x_i \,x_{i-1}\prod_{k\in
% [c-i]}\bigl(x_{i+k+1} + x_{i+k}\bigr).
% \end{multline*}
% Thus, for $s=1$ the Boolean function contains just \emph{one} degree-one term.

To get some intuition about the monomial structure, it is instructive to get
back to the simpler case $m=2^s+1$.
From~(\ref{eq:star}) it is evident that monomials associated with different
values of $i$ cannot be equal, as the minimum variable appearing in a monomial is $x_{i-1}$. Once we fix
$i$ with $1\leq i\leq c$, the number of monomials is equal to the number of
sets of the form
\begin{equation}
	\label{eq:subsets}
	S + s\;\cup\; [c-i]\setminus S\; \cup\; \sing{s-1} \qquad S\subseteq [c-i]
\end{equation}
that can be expressed by an odd number of values of $S$ (if you can
express the set in an even number of ways, they cancel out).
But such sets are in bijection
with the values $(v \ll s) \lor \lnot v \lor (1\ll s-1)$ as $v$ varies
among the words of $c-i$ bits. In a picture, we are looking at the columnwise logical or of the following
diagram, where the $b_j$'s are the bits of $v$, for convenience numbered from
the most significant:

\begin{figure}[h]
	\centering
	\includegraphics{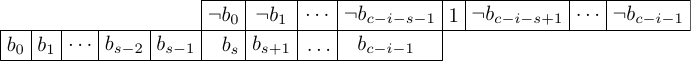}
\end{figure}

The first obvious observation is that if $s>c-i$ the two rows are nonoverlapping,
and they are not influenced by the one in position $s-1$. In this case, we
obtain all possible $2^{c-i}$ monomials.
%  Not only: we also have
% \[x_{i-1+s} x_{i-1}\sum_{S\subseteq[c-i]}\prod_{j\in S}x_{i+j+s}\prod_{j\in
% [c-i]\setminus S}x_{i+j} = x_{i-1+s} x_{i-1}\prod_{k\in [c-i]}\bigl(x_{i+k+s} +
% x_{i+k}\bigr),\] that is, the summand will assume exactly the type of structure we are trying to
% avoid.
% 
% \frameit{Seba}{Problem: the product above works for \emph{all values of $s$}. So
% there's something more than just the form. Maybe it must be a product where
% little or no cancellation occurs.}
% 
% 
More generally, such sets are all distinct iff
$s\geq(c-i+1)/2$, as in that case the values must differ either in the first $s$ or in the last $s-1$
bits: consequently, the number of monomials, in this case, is again $2^{c-i}$.
Minimizing $i$ and maximizing $c$ we obtain
$s\geq(w-s-1)/2$,
whence $s\geq(w-1)/3$. In this case, the monomials of $z_b$ are exactly
$2 + 2^{c-1}+2^{c-2}+\cdots + 1 = 2^c+1 = 2^{b-s}+1$ when $b\geq s$.

As $s$ moves down from $(w-1)/3$, we observe empirically more and more reduction
in the number of monomials with respect to the maximum possible $2^{b-s}+1$. When we reach $s=1$,
however, a radical change happens: the number of monomials grows as $2^{b/2}$.

\begin{theorem}
	\label{teo:carry3}
	The number of monomials of the Boolean function representing bit $b$ of $3x$
	is\footnote{Note that we are using Knuth's extension of \emph{Iverson's
			notation}~\cite{KnuTNN}:
		a Boolean expression between square brackets has value $1$ or $0$ depending on
		whether it is true or false, respectively.}
	\[(2+[\text{$b$ odd}])\cdot 2^{\lfloor b/2\rfloor} - 1.\]
\end{theorem}

We remark a surprising combinatorial connection: this is the number
of binary palindromes smaller than $2^b$, that is, A052955 in the ``On-Line
Encyclopedia of Integer Sequences''~\cite{OEIS}.

\begin{proof}

	When $s=1$, the different subsets
	in~(\ref{eq:subsets}) obtained when $S$ varies are in bijection
	with the values $v \lor \lnot (v \gg 1)$ as $v$ varies among the words of $c-i$
	bits. Again, we are looking at the logical or by columns of the
	following diagram, where the $b_j$'s are the bits of $v$ numbered from
	the most significant:
	\begin{center}
		\includegraphics{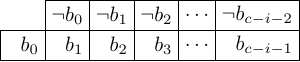}
	\end{center}

	Note that if there is a $b_j$ whose value is irrelevant, flipping will generate
	two monomials that will cancel each other.

	Let us consider now a
	successive assignment of values to the $b_j$'s, starting from $b_0$.
	We remark that as long as we assign ones, no assigned bit is irrelevant.
	As soon as we assign a zero, however, say to $b_j$, we have that the value of $b_{j+1}$ will
	no longer be relevant. To make $b_{j+1}$ relevant, we need to set $b_{j+2} =
		0$. The argument continues until the end of the word, so we can actually
	choose the value of $(c-i-j-1)/2$ bits, and only if $c-i-j-1$ is even
	(otherwise, $b_{c-i-1}$ has no influence).

	We now note that if we flip a bit $b_k$ that we were forced to set to zero,
	there are two possibilities: either we chose $b_{k-1}=1$, in which case we
	obtain a different monomial, or we chose $b_{k-1}=0$, in which case $b_k$ is
	irrelevant, but by flipping also $b_{k+1}$ we obtain once again the same monomial, so the two copies
	cancel each other.

	Said otherwise, monomials with an odd number of occurrences
	are generated either when all bits of $v$ are set to one, or when there is a
	string of ones followed by a suffix of odd length in which every other bit (starting from
	the first one) is zero. All in all, we have
	\[
		1 + \sum_{k=0}^{\bigl\lfloor \frac{c-i-1}2\bigr\rfloor} 2^k = 2^{\bigl\lceil
				\frac{c-i-1}2\bigr\rceil}
	\]
	possible monomials, where $2k+1$ is the length of the suffix. Adding up over all
	$i$'s, and adding the two degree-one monomials we have that the number of
	monomials of $y_b$ for $b=c+1>0$ is
	\[
		2 + \sum_{i=1}^{b-1}2^{\bigl\lceil \frac{b-i-2}2\bigr\rceil} =
		1 + 1 + \sum_{i=1}^{b-1}2^{\bigl\lceil \frac{b-i-2}2\bigr\rceil} =
		(2+[\text{$b$ odd}])\cdot 2^{\lfloor b/2\rfloor} - 1.
	\]
	The correctness for the case $b=0$ can be checked directly.
\end{proof}

We remark that in empirical tests the $3x$ scrambler performs very poorly:
thus, the excessive cancellation of monomials implied by the theorem above
has practical consequences.

\subsection{The \texttt{++} scrambler}

We will now examine the strong scrambler \texttt{++} introduced in
Section~\ref{sec:plusplus}. We choose two words $x$, $y$ from the state of the
linear engine and then $z = \rho_r(x+y)+x$, where $+$ denotes sum in $\Z/2^w\Z$.

Computing an ANF for the final Boolean functions appears to be a hard
combinatorial problem:
nonetheless, with this setup we know that the lowest bit will have
polynomial degree $w-r+1$, and we expect that the following bits
will have an increasing degree, possibly up to saturation. Symbolic computations
in low dimension show however that the growth is quite irregular.
The linear complexity of the lowest bits is large, as shown in Table~\ref{tab:ppdeg}, where
we display a theoretical estimate based on~(\ref{eq:ub}), assuming that on lower bits degree
increase at least by one at each bit (experimentally, it usually grows more quickly---see again Table~\ref{tab:ppdeg}).

\begin{table}
	\renewcommand{\arraystretch}{1.1}
	\renewcommand{\tabcolsep}{4pt}
	\centering
	\begin{tabular}{rrrrrr}
		{\small\xoroshiropp[128]} & {\small\xoshiropp[256]} & {\small\xoshiropp[512]} & {\small\xoroshiropp[1024]} & {\small\xoshiropp[128]} \\
		\hline
		% This data is for rotation 5-7-9 and phi-5-5
		$1\times 10^{36} $        & $3\times 10^{48} $      & $1\times 10^{68} $      & $9\times 10^{74}$          & $1\times 10^{27}$       \\
		$2\times 10^{36} $        & $2\times 10^{49} $      & $1\times 10^{69} $      & $2\times 10^{76}$          & $5\times 10^{27}$       \\
		$3\times 10^{36} $        & $1\times 10^{50} $      & $9\times 10^{69} $      & $4\times 10^{77}$          & $2\times 10^{28}$       \\
		$4\times 10^{36} $        & $4\times 10^{50} $      & $8\times 10^{70} $      & $1\times 10^{79}$          & $6\times 10^{28}$       \\
	\end{tabular}
	\caption{\label{tab:ppdeg}Approximate lower bound on the estimated linear complexity of the four lowest bits (the first line is bit 0) of
		generators using the \texttt{++} scrambler with
		parameters from Table~\ref{tab:scr64} and~\ref{tab:scr32}.}
\end{table}

This scrambler is potentially very fast, as it requires just three operations
and no multiplication, and it can reach a high polynomial degree, as it uses
$2w$ bits.\footnote{Symbolic computation suggests that this scrambler can reach
	only polynomial degree $2w-3$; while we have the bound $2w-1$ by
	Proposition~\ref{prop:wminus1}, proving the bound $2w-3$ is an open problem.}
Moreover, its simpler structure makes it attractive in hardware implementations.
However, the very regular structure of the \texttt{+} scrambler makes experimentally
\texttt{++} less effective on Hamming-weight dependencies.

As a basic heuristic, we suggest to choose a rotation parameter
$r\in[w/4\..3w/4]$ such that $r$ and $w-r$ are both prime (or at least odd), and they are not
equal to any of the shift/rotate parameters appearing in the generator (the
second condition being more relevant than the first one). Smaller values of $r$
will of course provide a higher polynomial degree, but too small values yield too short carry chains.
For $w=64$ candidates are $17$, $23$, $41$, and $47$; for $w=32$ one has
$13$ and $19$; for $w=16$ one has $5$ and $11$. In any case, a
specific combination of linear engine and scrambler should be
tested thoroughly.

As in the case of the \texttt{+} scrambler, if the underlying linear engine is
$d$-dimensionally equidistributed, the scrambler generator will be in general at
most $(d-1)$-dimensionally equidistributed (see Section~\ref{sec:ed}).

\subsection{The \texttt{**} scrambler}

We conclude our discussion with the strong scrambler \texttt{**} introduced in
Section~\ref{sec:starstar}. We will be discussing in detail the case with multiplicative
constants of the form $2^s+1$ and $2^t+1$, which is particularly fast (the $+$ symbol will denote
sum in $\Z/2^w\Z$ for the rest of this section).

Let $z = \rho_r(x\cdot (2^s+1)) \cdot (2^t
	+ 1)$. The $\min\{\?r,t\?\}$ lowest bits of $z$ are the
$\min\{\?r,t\?\}$ highest bits of $x\cdot (2^s+1)$.
To choose $s$, $r$, and $t$ we can leverage our previous knowledge of the
scrambler \texttt{*}.
We start by imposing that $s<t$, as choosing $s=t$ generates several duplicates
that reduce significantly the number of monomials in the ANF of the final
Boolean functions, whereas $t<s$ provably yields a lower minimum degree for the same $r$ (empirical
computations show also a smaller number of monomials). We also have to impose
$t<r $, for otherwise some bits or xor of pair of bits will have very low linear
complexity (polynomial degree one). So we
have to choose our parameters with the constraint $s<t<r$.
Since the degree of the lowest bit is $\max(1, w - r - s + 1)$, choosing $r =
	t+1$ maximizes the minimum degree across the bits. Moreover, we would like to
keep $s$ and $t$ as small as possible, to increase the minimum linear complexity
and also to make the scrambler faster.

Also in this case computing an ANF for the final Boolean functions appears to be a hard
combinatorial problem:
nonetheless, with this setup we know that the lowest bit will have (when
$r+s\leq w$) polynomial degree $w-r-s+1$, and we expect that the following bits
will have increasing degree up to saturation (which happens at degree $w-1$ by
Proposition~\ref{prop:wminus1}). Symbolic computations in low
dimension show some polynomial degree loss caused by the second multiplication unless $r=2t+1$; moreover, for that value of
$r$ the polynomial degree loss when combining bits is almost
absent.
Taking into consideration the bad
behavior of the multiplier $3$ highlighted by Theorem~\ref{teo:carry3}, we conclude that the best choice is $s=2$, $t=3$, and consequently $r=7$. These are the parameters reported in Table~\ref{tab:scr64}.
The linear complexity of the lowest bits is extremely large, as shown in Table~\ref{tab:ssdeg}.\footnote{Note that as we move
	towards higher bits the \texttt{++} scrambler will surpass the linear complexity of the \texttt{**} scrambler; the fact
	that the lower bits appear of lower complexity is due only to the fact that we use much larger rotations in the \texttt{++} case.}

\begin{table}
	\renewcommand{\arraystretch}{1.1}
	\renewcommand{\tabcolsep}{4pt}
	\centering
	\begin{tabular}{rrrrrr}
		{\small\xoroshiross[128]} & {\small\xoshiross[256]} & {\small\xoshiross[512]} & {\small\xoroshiross[1024]} & {\small\xoroshiross[64]} & {\small\xoshiross[128]} \\
		\hline
		% This data is for rotation 5-7-9 and phi-5-5
		$3\times 10^{37} $        & $2\times 10^{57} $      & $4\times 10^{75} $      & $1\times 10^{93}$          & $2\times 10^{18}$        & $8\times 10^{25}$       \\
		$4\times 10^{37} $        & $7\times 10^{57} $      & $3\times 10^{76} $      & $2\times 10^{94}$          & $3\times 10^{18}$        & $3\times 10^{26}$       \\
		$6\times 10^{37} $        & $3\times 10^{58} $      & $2\times 10^{77} $      & $3\times 10^{95}$          & $5\times 10^{18}$        & $1\times 10^{27}$       \\
		$7\times 10^{37} $        & $9\times 10^{58} $      & $2\times 10^{78} $      & $6\times 10^{96}$          & $6\times 10^{18}$        & $5\times 10^{27}$       \\
	\end{tabular}
	\caption{\label{tab:ssdeg}Approximate estimated linear complexity of the four lowest bit (the first line is bit 0) of
		generators using the \texttt{**} scrambler with
		parameters from Table~\ref{tab:scr64} and~\ref{tab:scr32}.}
\end{table}

At $32$ bits, however, tests show that this scrambler is not
sufficiently powerful for \xoroshiro[64], and Table~\ref{tab:scr32} reports
indeed different parameters:
the first multiplier is the constant used for the \texttt{*} scrambler, and the
second multiplier $2^t+1$ has been chosen so that bit $t$ is not set in the
first constant. Again, $r=2t+1$, following the same heuristic of the previous
case.

\section{Conclusions}

The combination of \xoroshiro/\xoshiro and suitable
scramblers provides a wide range of high-quality and fast solutions
for pseudorandom number generation. Parallax has embedded in their recently designed Propeller 2 microcontroller
\xoroshiross[128] and the $16$-bit \xoroshiropp[32];
\xoroshirop[116] is the stock generator of Erlang and \xoshiross[256]
is the stock generator of the popular embedded language Lua and of
GNU Fortran. \xoroshiropp[128] and \xoshiropp[256] are scheduled to
be included in Java 17 as part of JDK Enhancement Proposal 356.
Recently, the speed of
\xoshiross[128] has found application in cryptography~\cite{BFMFYL,GeREPMIQ}.

We believe that a more complete study of scramblers can shed some further light
on the behavior of such generators: the open problem is that of devising
a model explaining the elimination of Hamming-weight dependencies.
The main difficulty is that analyzing the Boolean functions representing each scrambled bit in isolation is not sufficient,
as Hamming-weight dependencies are generated by their collective behavior.

There are variants of the scramblers we discussed that do not use rotations: for
example, in the \texttt{++} and \texttt{**} scramblers the rotation can be
replaced by xoring $x$ with $x \gg r$, as also this operation will increase the
linear complexity of the lower bits. For contexts in which rotations are not
available or too expensive, one might explore the possibility of using \xorshift
generators scrambled with such variants.

% On a more theoretical side, it would be interesting to find a closed form
% for the number of monomials of the scrambler \texttt{*} with constant $(2^s+1)$
% when $s>1$, thus extending Theorem~\ref{teo:carry3}, or an ANF for
% the Boolean functions of the \texttt{++} and \texttt{**} scramblers.

% For more complex matrices, word polynomials might be based on more general
% approaches, such as \emph{quasideterminants}~\cite{GeRDMNR}. If a word
% polynomial predicts exactly the generator it might be interesting to
% tune the parameters of the generator using the figure of merit introduced
% in~\cite{NieRNGQMCM}.

There is a vast literature on filtered LFSR that might be used
to prove aspects we approached only with symbolic small-state computations.
For example, in~\cite{KLKIBLCKOFG} the authors prove a lower bound on the linear
degree of a Boolean function made of a single very specific monomial, something
for which we just argued based on measurements made using the Berlekamp--Massey algorithm.
In~\cite{BrSANF} the authors try to provide closed forms or even ANFs when
the argument of a Boolean function is multiplied or summed with a constant, which
might be a starting point for a closed form for the \texttt{**} scrambler.
% Finally, several concepts developed in the context of filtered LFSR, 
% such as \emph{algebraic immunity} and \emph{Bent functions}
% might prove to be useful in choosing scrambler parameters, or defining new ones.

In general, it is an interesting open problem to correlate explicitly the
monomial structure of a Boolean function in ANF with its resilience to linearity
tests. Intuitively, recalling~(\ref{eq:neg}), one sees that besides large-degree
monomials one needs small-degree monomials to make the tests ``perceive'' the
increase in linear complexity at the right time.

\begin{acks}
	The authors would like to thank Parallax developers Chip Gracey, Evan Hillas
	and Tony Brewer for their interest, enthusiasm and proofreading,
	Pierre L'Ecuyer for a number of suggestions that significantly improved the
	quality of the presentation, Raimo Niskanen of the Erlang/OTP team
	and Nat Sothanaphan for several useful discussions, Guy Steele for stimulating correspondence
	and for suggesting to include data-dependency diagrams,
	Robert H. Lewis for the unbelievable speed of Fermat~\cite{LewF}, and the Sage authors
	for a wonderful tool~\cite{Sage}.
\end{acks}

% PIERRE: reference to ``dynamic creation'' and ``nonempirical test'', volume
% for Knuth

\bibliography{biblio}

\end{document}